%% file: main.tex
\documentclass[a4paper]{llncs}
\usepackage[usenames]{color}
\usepackage{fit_tr}
\usepackage{hyperref}
\usepackage{latexsym}
\usepackage{amsmath}
\usepackage[linesnumbered,vlined,ruled]{algorithm2e}
\usepackage{amssymb}
\usepackage{tikz}
\usepackage{mathdots}
\usepackage{paralist}
\usetikzlibrary{arrows,automata}
\usepackage{multirow}
\usepackage{wrapfig}

\input{macros}

\title{Nested Antichains for WS1S}
\author{
  Tom\'{a}\v{s} Fiedor, 
  Luk\'{a}\v{s} Hol\'{i}k,
  Ond\v{r}ej Leng\'{a}l, and
  Tom\'{a}\v{s} Vojnar
}

\institute{
  {FIT, Brno University of Technology, IT4Innovations Centre of Excellence, Czech~Republic}
}

\begin{document} 
%%%%%%%%%%%%%%%%%%%%%%%%%%%%%%%%%%%%%%%%%%%%%%%%%%%%%%%%%%%%%%%%%%%%%%%%%%%%%%%%

\booktitle{Nested Antichains for WS1S} % Title
{} % Subtitle
{FIT BUT Technical Report Series}
{Tom\'{a}\v{s} Fiedor, Luk\'{a}\v{s} Hol\'{\i}k,\\[2mm]
Ond\v{r}ej Leng\'{a}l, and Tom\'{a}\v{s}~Vojnar} % Authors
{Technical Report No. FIT-TR-2014-06\\[2mm]
 Faculty of Information Technology, Brno University of Technology}
{Last modified: \today}

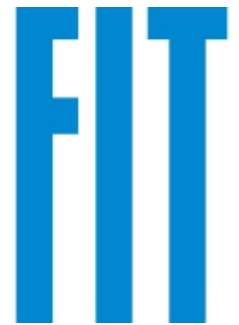
\eject

\pagestyle{empty}
\noindent\textbf{NOTE:} This technical report contains an extended version of
a~paper with the same name accepted to TACAS'15.

\eject

% Decrease the page counter by 2
\addtocounter{page}{-2}

%%%%%%%%%%%%%%%%%%%%%%%%%%%%%%%%%%%%%%%%%%%%%%%%%%%%%%%%%%%%%%%%%%%%%%%%%%%%%%%
% The regular LNCS document begins now...
% (Add \thispagestyle{plain} behind \maketitle!)
%%%%%%%%%%%%%%%%%%%%%%%%%%%%%%%%%%%%%%%%%%%%%%%%%%%%%%%%%%%%%%%%%%%%%%%%%%%%%%%

\pagestyle{plain}

%\pdfpagesattr{/CropBox [92 62 523 728]} % LNCS page made big 

%%%%%%%%%%%%%%%%%%%%%%%%%%%%%%%%%%%%%%%%%%%%%%%%%%%%%%%%%%%%%%%%%%%%%%%%%%%%%%%%
\maketitle

\begin{abstract}
We propose a novel approach for coping with alternating quantification as the
main source of nonelementary complexity of deciding WS1S formulae.
Our approach is applicable within the state-of-the-art automata-based WS1S
decision procedure implemented, e.g.\ in MONA.
The way in which the standard decision procedure processes quantifiers involves
determinization, with its worst case exponential complexity, for every
quantifier alternation in the prefix of a~formula.
Our algorithm avoids building the deterministic automata---instead, it
constructs only those of their states needed for (dis)proving validity of the
formula.
It uses a symbolic representation of the states, which have a~deeply nested
structure stemming from the repeated implicit subset construction, and prunes
the search space by a nested subsumption relation, a generalization of the one
used by the so-called antichain algorithms for handling nondeterministic
automata.
We have obtained encouraging experimental results, in some cases outperforming
MONA by several orders of magnitude.
\end{abstract}

\vspace{-0.0mm}
\section{Introduction}\label{sec:intro}
\vspace{-0.0mm}
%%%%%%%%%%%%%%%%%%%%%%%%%%%%%%%%%%%%%%%%%%%%%%%%%%%%%%%%%%%%%%%%%%%%%%%%%%%%%%%%

Weak monadic second-order logic of one successor (WS1S) 
is a powerful, concise, and decidable logic for describing regular properties of finite words.
Despite its nonelementary worst case complexity~\cite{meyer-lc-72},
it has been shown useful in numerous applications. 
Most of the successful applications were due to the tool MONA \cite{monapaper},
which implements a finite automata-based decision procedure for WS1S and WS2S
(a generalization of WS1S to finite binary trees).
%, first implemented for WS1S in \cite{glenn-wia-96}.
The authors of MONA list a multitude of its diverse applications~\cite{monamanual},
ranging from software and hardware verification through controller synthesis to computational linguistics, and further on.
% A long list of applications listed in \cite{monamanual} ranges over software and hardware verification, controller synthesis, computational linguistics, and many others. 
Among more recent applications, verification of pointer programs and deciding related logics~\cite{strand1,strand2,adam,hip/sleek,jahob} can be mentioned, as well as synthesis from regular specifications~\cite{regsy}.
MONA is still the standard tool and the most common choice when it comes to deciding WS1S/WS2S. 
There are other related automata-based tools that are more recent, 
such as jMosel~\cite{jmosel} for a logic M2L(Str),
and other than automata-based approaches, such as~\cite{ganzow:new}.
They implement optimizations that allow to outperform MONA on some benchmarks, however, none provides an evidence of being consistently more efficient. 
%% 
%%  % To the best of our knowledge, MONA is still the best available tool for deciding WS1S (\wsks),
%%  % and other than automata-based approaches, such as \cite{ganzow:new}, have not matured to the same level of usability.
%%  To the best of our knowledge, MONA is still the best available tool for deciding WS1S (\wsks).
%%  There are other related automata-based tools, such as jMosel~\cite{jmosel} for a related logic M2L(Str),
%%  and other than automata-based approaches, such as \cite{ganzow:new}, but none of them have matured to the same level of usability as MONA.
%%  %of practical applicability as MONA.
%%  %To the best of our knowledge, other than automata-based approaches, such as \cite{ganzow:new}, have not matured to the same level of practical applicability as MONA.
Despite many optimizations implemented in MONA and the other tools, the worst case complexity of the problem sometimes strikes back. 
%back and limits applicability of the tool.
%
%To the best of our knowledge, other than automata-based approaches, such as \cite{ganzow:new}, have not matured to the same level of practical applicability as MONA.
%Even though MONA is highly optimised, the worst case complexity of the problem sometimes strikes back and limits applicability of the tool. 
%
%Even though MONA is highly optimised, the worst case complexity of the problem sometimes strikes. 
Authors of methods using the translation of their problem to WS1S/WS2S 
are then forced to either find workarounds to circumvent the complexity blowup, such as in~\cite{strand2},
or, often restricting the input of their approach,
give up translating to WS1S/WS2S altogether~\cite{kuncak:trex}.
% \td{what is the relation between wsks, WMSO, Adams paper?}
% \td{other logics like mso, presburger?}

The decision procedure of MONA works with deterministic automata;
it uses determinization extensively and relies on minimization of deterministic automata to suppress the complexity blow-up.
However, the worst case exponential complexity of determinization often significantly harms the performance of the tool.
Recent works on efficient methods for handling nondeterministic automata suggest a way of alleviating this problem,
in particular works on efficient testing of language inclusion and universality of finite automata \cite{doyen:antichain,wulf:antichains,abdulla-tacas-10} and size reduction \cite{bustan:simulation,abdulla:computing} based on a~simulation relation.
Handling nondeterministic automata using these methods, while avoiding determinization, has been shown to provide great efficiency improvements in 
\cite{bouajjani:antichain} (abstract regular model checking) and also \cite{habermehl:forest} (shape analysis).
In this paper, we make a major step towards building the entire decision procedure of WS1S on nondeterministic automata using similar techniques.
We propose a generalization of the antichain algorithms of~\cite{doyen:antichain}
that addresses the main bottleneck of the automata-based decision procedure for WS1S, 
which is also the source of its nonelementary complexity: 
elimination of alternating quantifiers on the automata level.

%In this paper, we make the first significant step in addressing the complexity
%explosion with the use of nondeterministic automata and efficient techniques for
%manipulating them.
%Nondeterministic automata have been already shown to provide great efficiency benefits over deterministic ones in some applications (such as shape analysis),
%mainly thanks to the so-called antichain algorithms for testing language inclusion and universality \cite{doyen:antichain,wulf:antichains,abdulla:when} and size reduction methods \cite{bustan:simulation} based on simulation.  
%In this paper, 
%we devise a generalization of the antichain algorithms to  
%address the main bottleneck of the automata-based decision procedure for WS1S, 
%which is also the source of its nonelementary complexity: 
%elimination of alternating quantifiers on the automata level.

More concretely, the automata-based decision procedure translates the input WS1S formula into a finite word automaton such that its language represents exactly all models of the formula. The automaton is built in a bottom-up manner according to the structure of the formula, starting with predefined atomic automata for literals and applying a~corresponding automata operation for every logical connective and quantifier ($\land,\lor,\neg,\exists$).
The cause of the nonelementary complexity of the procedure can be explained on an example formula of the form
$\varphi' = \exists X_{m}\forall X_{m-1} \ldots \forall X_{2} \exists X_1: \varphi_0$.
The universal quantifiers are first replaced by negation and existential quantification, 
which results in
$\varphi = \exists X_{m}\neg \exists X_{m-1} \ldots \neg \exists X_{2} \neg\exists X_1: \varphi_0$.  
%The algorithm then builds one by one automata for the sub-formulae $\varphi_0,\varphi_0^\sharp,\ldots,\varphi_{m-1},\varphi_{m-1}^\sharp$ of $\varphi$ where for $0\leq  i < m$,
The algorithm then builds a sequence of automata for the sub-formulae $\varphi_0,\varphi_0^\sharp,\ldots,\varphi_{m-1},\varphi_{m-1}^\sharp$ of $\varphi$ where for $0\leq  i < m$,
$\varphi_i^{\sharp} = \exists X_{i+1}:\varphi_i$, and $\varphi_{i+1} = \neg\varphi_i^{\sharp}$. 
Every automaton in the sequence is created from the previous one by applying the automata operations corresponding to negation or elimination of the existential quantifier,
%Elimination of the existential quantifier introduces nondeterminism.
the latter of which may introduce nondeterminism.
Negation applied on a nondeterministic automaton may then yield an exponential blowup: 
given an automaton for $\psi$, the automaton for $\neg\psi$ is constructed by the classical automata-theoretic construction consisting of determinization by the subset construction followed by swapping of the sets of final and non-final states. 
%The subset construction is well-known to be exponential in the worst case,
The subset construction is exponential in the worst case.
The worst case complexity of the procedure run on $\varphi$ is then a tower of exponentials with one level for every quantifier alternation in $\varphi$;
note that we cannot do much better---this non-elementary complexity is an
inherent property of the problem.

Our new algorithm for processing alternating
quantifiers in the prefix of a~formula avoids the explicit determinization
%(and the associated exponential blowup) 
of automata in the classical procedure and significantly reduces the state space explosion associated with it.
% The optimisation \td{OL: novel approach} of the procedure proposed in this paper aims at suppressing the exponential blow-up that occurs when processing alternating quantifiers and negation as described above. 
It is based on a generalization of the antichain principle used for deciding universality and language inclusion of finite automata \cite{wulf:antichains,abdulla-tacas-10}.
It generalizes the antichain algorithms so that instead of being used to process only one level of the chain of automata, 
it processes the whole chain of quantifications with $i$ alternations on-the-fly.
This leads to working with automata states that are sets of sets of sets \dots{} of states of the automaton representing $\varphi_0$ of the nesting depth $i$ (this corresponds to $i$ levels of subset construction being done on-the-fly).
The algorithm uses nested symbolic terms to represent sets of such automata states and a generalized version of antichain subsumption pruning which descends recursively down the structure of the terms while pruning on all its levels.

Our nested antichain algorithm can be in its current form used only to process a~quantifier prefix of a formula, 
after which we return the answer to the validity query, 
but not an automaton representing all models of the input formula. 
That is, we cannot use the optimized algorithm for processing inner negations and alternating quantifiers which are not a part of the quantifier prefix. %, loosing many opportunities for optimisation.
%We must either use the standard automata algorithm for that or translate the formula to the prenex form, which both means loosing oportunities for optimisation.
%We are currently not able to use it when processing the inner sub-formula which contains quantifiers that do not belong to the quantifier prefix of the formula. 
%Therefore, we either (1) generate the automaton for the inner formula using the standard procedure and run the nested antichain algorithm on the resulting automaton,
%or (2) run our algorithm on a formula transformed into the prenex form (a sequence of quantifiers followed by a quantifier free formula). 
%Both ways have drawbacks. The first way leads to loosing advantages of nondeterministic automata when processing the inner part of the formula. 
%The other wastes opportunities for optimisations that MONA uses when processing boolean combinations of quantified sub-formulae.
However, despite this and the fact that our implementation is far less mature than that of MONA, 
our experimental results still show significant improvements over its performance, 
especially in terms of generated state space. 
We consider this a~strong indication that using techniques for nondeterministic automata to decide WS1S (and WS$k$S) is highly promising.
There are many more opportunities of improving the decision procedure based on nondeterministic automata, by using
techniques 
%for handling nondeterministic automata 
such as simulation relations or bisimulation up-to congruence \cite{bonchi:congruence}, and applying them to process not only the quantifier prefix, but all logical connectives of a formula. We consider this paper to be the first step towards a~decision procedure for WS1S/WS$k$S with an entirely different scalability than the current state-of-the-art. 

\paragraph{Plan of the paper.}
We define the logic WS1S in Section~\ref{sec:wsks}.
In Sections~\ref{sec:fa} and \ref{sec:dec_proc},
we introduce finite word automata and describe the classical decision procedure for WS1S based on finite word automata.
In Section~\ref{sec:dec_proc_ws1s}, we introduce our method for dealing with alternating quantifiers.
%We then generalise the described procedure for WS1S to a tree automata-based decision procedure for \wsks.
Finally, we give an experimental evaluation and conclude the paper in Sections~\ref{sec:experiments} and \ref{sec:conclusion}.

% such as a use of other sofisticated techniques for handling nondeterministic automata (e.g. simulation relations \cite{} or bisimulation up to congruence) generalising the nested antichain principle to boolean combinations of automata (e.g. simulation relations \cite{}) or generalising the antichain algorithm for handling disjunction and conjuction (similarly as in the algorithm for computing language inclusion)

%% 
%% 
%% 
%% 
%% \td{OL: fill}
%% \td{OL: MONA}
%% \td{OL: guys using MONA}
%% \td{OL: Glenn, Gasarch}
%% \td{OL: mention complexity (Meyer)}
%% 
%% An interesting theoretical result~\cite{meyer-lc-72} is that the computational
%% complexity of \wsks{} is \textbf{NONELEMENTARY}, that is, given a Turing machine
%% $\M$ deciding satisfiability of \wsks{} formulae, for any $u \geq 0$ there are
%% infinitely many $n$ for which a~computation of $\M$ for some sentence of length
%% $n$ requires at least
%% \begin{equation*}
%%   \underset{u}{\underbrace{2^{2^{\cdot^{\cdot^{\cdot^{2^n}}}}}}}\quad\mathrm{steps.}
%% \end{equation*}

%%%%%%%%%%%%%%%%%%%%%%%%%%%%%%%%%%%%%%%%%%%%%%%%%%%%%%%%%%%%%%%%%%%%%%%%%%%%%%%%
\vspace{-0.0mm}
\section{WS1S}\label{sec:wsks}
\vspace{-0.0mm}
%%%%%%%%%%%%%%%%%%%%%%%%%%%%%%%%%%%%%%%%%%%%%%%%%%%%%%%%%%%%%%%%%%%%%%%%%%%%%%%%

In this section we introduce the \emph{weak monadic second-order logic of one
successor} (WS1S).
We introduce only its minimal syntax here, 
for the full standard syntax and a~more thorough introduction,
see Section~3.3 in~\cite{tata}.
%In the following, let us fix $k \geq 1$. 

WS1S is a monadic second-order logic over the universe of discourse $\nat_0$.
This means that the logic allows second-order
\emph{variables}, usually denoted using upper-case letters $X, Y, \dots$, that range
over finite subsets of $\nat_0$, e.g.\ $X = \{0, 3, 42\}$.
Atomic formulae are of the form
\begin{inparaenum}[(i)]
\item  $X \subseteq Y$,
\item  $\singof{X}$,
\item  $X = \{0\}$, and
\item  $X = Y + 1$,
\end{inparaenum}
where $X$ and $Y$ are variables.
The atomic formulae are interpreted in turn as
\begin{inparaenum}[(i)]
\item  standard set inclusion,
\item  the singleton predicate,
\item  $X$ is a singleton containing 0, and
\item  $X = \{x\}$ and $Y = \{y\}$ are singletons and $x$ is the successor of $y$, i.e.\ $x = y + 1$.
\end{inparaenum}
Formulae are built from the atomic formulae using the logical connectives
$\wedge, \vee, \neg$, and the quantifier $\exists X$ (for a second-order variable
$X$).
%As usual, more atomic formulae and Boolean operators can be added but the
%above ones are sufficient to obtain the full expressive power of the logic.

%In the further text, we consider only second-order variables; a first-order
%variable $x$ can be expressed using a second-order variable $X$ augmented with
%the constraint $\singof{X}$, atomic formulae of the form $x = \succof{i}{y}$
%are substituted with formulae of the form $X = \succof{i}{Y}$ interpreted as $X
%= \succof{i}{Y} \iffdef \exists x, y : X = \{x\} \land Y = \{y\} \land x =
%\succof{i}{y}$, and atomic formulae of the form $x = \epsilon$ are substituted
%with formulae of the form $X = \epsilon$ interpreted as $X = \epsilon \iffdef
%\exists x : X = \{x\} \land x = \epsilon$.

Given a WS1S formula $\varphi(X_1, \dots, X_n)$ with free
variables $X_1, \dots, X_n$, the assignment $\rho = \{X_1 \mapsto S_1, \dots,
X_n \mapsto S_n\}$, where $S_1, \dots, S_n$ are finite subsets of
$\nat_0$, \emph{satisfies} $\varphi$, written as $\rho \models \varphi$, if the
formula holds when every variable $X_i$ is replaced with its corresponding
value $S_i = \rho(X_i)$.
We say that $\varphi$ is \emph{valid}, denoted as $\models \varphi$, if it is satisfied by
all assignments of its free variables to finite subsets of $\nat_0$.
Observe the limitation to \emph{finite} subsets of $\nat_0$ (related to the
adjective \emph{weak} in the name of the logic); a WS1S formula can indeed
only have finite models (although there may be infinitely many of them).

\vspace{-0.0mm}
\section{Preliminaries and Finite Automata}\label{sec:fa}
\vspace{-0.0mm}
%%%%%%%%%%%%%%%%%%%%%%%%%%%%%%%%%%%%%%%%%%%%%%%%%%%%%%%%%%%%%%%%%%%%%%%%%%%%%%%

For a set $D$ and a set $\bbS \subseteq \powerset{D}$ we use $\downclngen{\bbS}$ to denote the
\emph{downward closure} of $\bbS$, i.e.\ the set $\downclngen{\bbS} = \{R \subseteq D \mid \exists S \in \bbS : R \subseteq
S\}$, and $\upclngen{\bbS}$ to denote the \emph{upward closure} of $\bbS$, i.e.\ the set $\upclngen{\bbS} =
\{R \subseteq D \mid \exists S \in \bbS : R \supseteq S\}$.
The set $\bbS$ is in both cases called the set of \emph{generators} of $\upclngen{\bbS}$ or $\downclngen{\bbS}$ respectively.
%For instance, if $D = \{1,2,3\}$, the downward closure
%$\downclgen{\{1,2\},\{3\}} = \{\emptyset, \{1\}, \{2\}, \{1,2\}, \{3\}\}$ and
%the upward closure $\upclgen{\{1,2\},\{3\}} = \{\{1,2\}, \{1,3\}, \{2,3\},
%\{3\}, D\}$.
A set $\bbS$ is \emph{downward closed} if it equals its downward closure, $\bbS =
\downclngen{\bbS}$, and \emph{upward closed} if it equals to its upward closure, $\bbS
= \upclngen{\bbS}$.
The \emph{choice} operator $\choice$ (sometimes also called the unordered Cartesian product) is an operator that, given a set of sets
$\bbD = \{D_1, \dots, D_n\}$, returns the set of all sets $\{d_1, \dots, d_n\}$
obtained by taking one element $d_i$ from every set $D_i$.
Formally,
\begin{equation}
\choice{\bbD} = \bigl\{\{d_1, \dots, d_n\} \mid (d_1, \dots, d_n) \in \prod_{i=1}^n D_i\bigr\}
%\choice{\bbD} = \left\{\{d_1, \dots, d_n\} \rlmid (d_1, \dots, d_n) \in \prod_{i=1}^n D_i\right\}
%\choice{\bbD} = \left\{\{d_1, \dots, d_n\} \rlmid (d_1, \dots, d_n) \in \prod_{1 \leq i \leq n} D_i\right\}
%\choice{\bbD} = \left\{\{d_1, \dots, d_n\} \rlmid (d_1, \dots, d_n) \in \prod \bbD\right\}
%\choice{\bbD} = \left\{\{d_1, \dots, d_n\} \rlmid (d_1, \dots, d_n) \in \prod_{1\leq i\leq n} \{D_1, \dots, D_n\}\right\}
\end{equation}
where $\prod$ denotes the Cartesian product.
% where $\prod \bbD$ is the Cartesian product of the set of sets $\bbD$ ordered
% in an arbitrary way.
Note that for a set $D$, $\choice{\{D\}}$ is the set of all
singleton subsets of $D$, i.e.\ $\choice{\{D\}} = \{\{d\} \mid d \in D\}$.
Further note that if any $D_i$ is the empty set $\emptyset$, the result is
$\choice{\bbD} = \emptyset$.

Let $\bbX$ be a set of variables.
A \emph{symbol} $\tau$ over $\bbX$ is a mapping of all variables in $\bbX$ to either 0
or 1, e.g.\ $\tau = \{X_1 \mapsto 0, X_2 \mapsto 1 \}$ for $\bbX = \{X_1, X_2\}$.
An \emph{alphabet} over $\bbX$ is the set of all symbols over $\bbX$, denoted
as $\Sigma_\bbX$.
For any $\bbX$ (even empty), we use $\zerosymb$ to denote the symbol which maps all
variables from $\bbX$ to 0, $\zerosymb \in \Sigma_{\bbX}$.

A (nondeterministic) \emph{finite} (word) \emph{automaton} (abbreviated
as FA in the following) over a~set of variables $\bbX$ is a~quadruple
$\A = (Q, \Delta, I, \finst{})$ where $Q$ is a~finite set of states, $I \subseteq Q$
is a~set of \emph{initial} states, $\finst{} \subseteq Q$ is a~set of \emph{final} states, and
$\Delta$ is a~set of transitions
of the form $(p, \tau, q)$ where $p, q \in Q$ and
$\tau \in \Sigma_\bbX$.
We use $p \ltr{\tau} q \in \Delta$ to denote that $(p, \tau, q) \in \Delta$.
Note that for an FA $\A$ over $\bbX = \emptyset$, $\A$ is a unary FA with the
alphabet $\Sigma_{\bbX} = \{\zerosymb\}$.

% Let $\A = (Q, \Delta, I, \finst{})$ be an FA over $\bbX$.
A \emph{run} $r$ of $\A$ over a word $w = \tau_1 \tau_2 \dots \tau_n \in \Sigma_\bbX^{*}$ from the
state $p \in Q$ to the state $s \in Q$ is a sequence of states $r = q_0 q_1 \dots
q_n \in Q^{+}$ such that $q_0 = p$, $q_n = s$ and for all $1 \leq i \leq n$ there
is a transition $q_{i-1} \ltr{\tau_i} q_i$ in $\Delta$.
If $s \in \finst{}$, we say that $r$ is an \emph{accepting run}.
We write $p \lrun{w} s$ to denote that there exists a~run from the state $p$ to
the state $s$ over the word $w$.
The \emph{language} accepted by a~state $q$ is defined by $\lang_{\A}(q)
= \{w\st q \lrun{w} q_f, q_f \in \finst{}\}$, while the language of a~set of states $S
\subseteq Q$ is defined as $\lang_{\A}(S) = \bigcup_{q \in S} \lang_{\A}(q)$.
When it is clear which FA $\A$ we refer to, we only write $\lang(q)$ or
$\lang(S)$.
The language of $\A$ is defined as $\lang(\A) = \lang_{\A}(I)$.
We say that the state $q$ accepts $w$ and that the automaton $\A$ accepts $w$ to
express that $w \in \lang_{\A}(q)$ and $w\in \lang(\A)$ respectively.
We call a~language $L \subseteq \Sigma_\bbX^{*}$ \emph{universal} iff $L = \Sigma_\bbX^{*}$.

%For a state $s \in Q$, a transition relation $\Delta$ and a symbol $\tau$, we
%define
%$\fwdof{\Delta, \tau}{s} = \{t \mid s \ltr{\tau} t \in \Delta\}$,
%$\bwdof{\Delta, \tau}{s} = \{t \mid t \ltr{\tau} s \in \Delta\}$.
For a set of states $S \subseteq Q$, we define
\begin{align*}
\postof{\Delta, \tau}{S} & {}= \bigcup_{s \in S} \{t \mid s \ltr{\tau} t \in \Delta\}, \\
\preof{\Delta, \tau}{S}  & {}= \bigcup_{s \in S} \{t \mid t \ltr{\tau} s \in \Delta\},~\text{and} \\
\cpreof{\Delta, \tau}{S} & {}= \{t \mid \postof{\Delta, \tau}{\{t\}}\subseteq S\}.
\end{align*}

% $\postof{\Delta, \tau}{S} = \bigcup_{s \in S} \{t \mid s \ltr{\tau} t \in \Delta\}$,
% $\preof{\Delta, \tau}{S}  = \bigcup_{s \in S} \{t \mid t \ltr{\tau} s \in \Delta\}$, and
% %$\cpreof{\Delta, \tau}{S} = \{t \mid \forall s . t \ltr{\tau} s \in \Delta : s \in S\}$.
% $\cpreof{\Delta, \tau}{S} = \{t \mid \postof{\Delta, \tau}{\{t\}}\subseteq S\}$.

%\td{MUCH MORE PROC BEAUTIFUL:}
%For a set of states $S \subseteq Q$, we define
%$\postof{\Delta, \tau}{S} = \Delta \cap S\times\{\tau\}\times Q$,
%$\preof{\Delta, \tau}{S}  = \Delta \cap Q\times\{\tau\}\times S$, and
%$\cpreof{\Delta, \tau}{S} = \{t \mid \forall s : t \ltr{\tau} s \in \Delta \implies s \in S\}$.

%We now define several operations on FAs.
%Let $\A_1 = (Q_1, \Delta_1, I_1, \finst{1})$ and $\A_2 = (Q_2, \Delta_2, I_2,
%\finst{2})$ be FAs over $\bbX$.
%The \emph{union} of $\A_1$ and $\A_2$ is the FA $\A_{\cup}$ over $\bbX$ defined
%as $\A_{\cup} = (Q_1 \cup Q_2, \Delta_1 \cup \Delta_2, I_1 \cup I_2,
%\finst{1} \cup \finst{2})$.
%Further, the \emph{intersection} of $\A_1$ and $\A_2$ is the FA $\A_{\cap}$
%over $\bbX$ defined as $\A_{\cap} = (Q_1 \times Q_2, \Delta_{\cap}, I_1 \cap
%I_2, \finst{\cap})$ where
%$\Delta_{\cap} = \left\{(p_1,p_2) \ltr{\tau} (q_1, q_2) \rlmid p_1 \ltr{\tau}
%q_1 \in \Delta_1 \land p_2 \ltr{\tau} q_2 \in \Delta_2 \right\}$
%and
%$\finst{\cap} = \finst{1} \times Q_2 \cup Q_1 \times \finst{2}$.
%Observe that $\langof{\A_{\cup}} = \langof{\A_1} \cup \langof{\A_2}$ and
%$\langof{\A_{\cap}} = \langof{\A_1} \cap \langof{\A_2}$.

The \emph{complement} of $\A$ is the automaton $\A_{\C} =
%(\powerset{Q}, \Delta_{\C}, \{I\}, \{P \subseteq Q \mid P \cap F = \emptyset\})$ where
(\powerset{Q}, \Delta_{\C}, \{I\}, \downclgen{Q\setminus F})$ where
%$\Delta = \left\{P \ltr{\tau} Q \rlmid Q = \postof{\Delta_1, \tau}{P}\right\}$;
$\Delta_{\C} = \left\{P \ltr{\tau} \postof{\Delta, \tau}{P} \rlmid P \subseteq Q\right\}$;
this corresponds to the standard procedure that
first determinizes $\A$ by the subset construction and then swaps its sets of final and non-final states,
and $\downclgen{Q\setminus F}$ is the set of all subsets of $Q$ that do not contain a final state of $\A$. 
The language of $\A_\C$ is the complement of the language of $\A$,
i.e.\ $\langof{\A_{\C}} = \overline{\langof{\A}}$.

%Let $\A = (Q, \Delta, I, \finst{})$ be an FA over $\bbX$ and $X$ be a variable.
%The \emph{cylindrification} of $\A$ by $X$
%\td{OL: is this how the cool guys
%say it? We are cool so we need to say it exactly how the cool guys say it...}
%is either $\A$ in the case $X \in \bbX$, or an FA obtained from $\A$
%by substituting every transition $p \ltr{\tau} q$ in $\A$ with a pair of
%transitions $p \ltr{\tau \cup \{X : 0\}} q$ and $p \ltr{\tau \cup \{X : 1\}}
%q$.
%Notice that the size of the alphabet of the resulting FA for the latter case
%doubled because every mapping $\tau$ yielded a pair of mappings that copy the
%mappings of $\tau$ and map $X$ to 0 and 1 respectively.

For a set of variables $\bbX$ and a variable $X$, the \emph{projection} of $X$
from $\bbX$, denoted as $\projOfFrom{X}{\bbX}$, is the set $\bbX \setminus
\{X\}$.
For a symbol $\tau$, the projection of $X$
from $\tau$, denoted $\projOfFrom{X}{\tau}$, is obtained from $\tau$ by
restricting $\tau$ to the domain $\projOfFrom{X}{\bbX}$.
For a transition relation $\Delta$, the projection of $X$ from $\Delta$,
denoted as $\projOfFrom{X}{\Delta}$, is the transition relation 
%$\left\{p \ltr{\tau'} q \rlmid \tau' = \projOfFrom{X}{\tau} \land p \ltr{\tau} q \in \Delta\right\}$.
$\Big\{p \ltr{\projOfFrom{X}{\tau}} q \mid p \ltr{\tau} q \in \Delta\Big\}$.
%

%%%%%%%%%%%%%%%%%%%%%%%%%%%%%%%%%%%%%%%%%%%%%%%%%%%%%%%%%%%%%%%%%%%%%%%%%%%%%%%%
\vspace{-0.0mm}
\section{Deciding WS1S with Finite Automata}\label{sec:dec_proc}
\vspace{-0.0mm}
%%%%%%%%%%%%%%%%%%%%%%%%%%%%%%%%%%%%%%%%%%%%%%%%%%%%%%%%%%%%%%%%%%%%%%%%%%%%%%%%

The classical decision procedure for WS1S~\cite{buchi59} (as described in Section~3.3
of~\cite{tata}) is based on a~logic-automata
connection and decides validity (satisfiability) of a WS1S formula
$\varphi(X_1, \dots, X_n)$ by constructing the FA $\A_\varphi$ over $\{X_1, \dots, X_n\}$
which recognizes encodings of exactly the models of $\varphi$.
The automaton is built in a bottom-up manner, according to the structure of $\varphi$, 
starting with predefined atomic automata for literals and applying a corresponding automata operation for every logical connective and quantifier ($\land,\lor,\neg,\exists$).
Hence, for every sub-formula $\psi$ of $\varphi$, the procedure will compute the automaton $\A_\psi$ such that $\langof{\A_\psi}$ represents exactly all models of $\psi$, terminating with the result $\A_\varphi$.

The alphabet of $\A_\varphi$ consists of all symbols over the set $\bbX =
\{X_1, \dots, X_n\}$ of free variables of $\varphi$
(for $a,b \in \{0,1\}$ and $\bbX = \{X_1, X_2\}$, we use
\bintrack{X_1}{X_2}{a}{b} to denote the symbol $\{X_1 \mapsto a, X_2 \mapsto b\}$).
A~word $w$ from the language of $\A_\varphi$ is a sequence of these symbols,
e.g.\ \bintrack{X_1}{X_2}{\epsilon}{\epsilon},
\bintrack{X_1}{X_2}{011}{101}, or \bintrack{X_1}{X_2}{01100}{10100}.
% for $\bbX = \{X_1, X_2\}$.
We denote the $i$-th symbol of $w$ as $w[i]$, for $i \in \nat_0$.
An assignment $\rho:\bbX \to \powerset{\nat_0}$ mapping free variables $\bbX$ of $\varphi$ to subsets of $\nat_0$ is encoded into a
word $w_{\rho}$ of symbols over $\bbX$ in the following way:
$w_{\rho}$ contains $1$ in the $j$-th position of the row for $X_i$ iff $j \in X_i$ in $\rho$.
Formally, for every $i \in \nat_0$ and $X_j \in \bbX$, if $i \in \rho(X_j)$, then $w_{\rho}[i]$ maps $X_j \mapsto 1$.
On the other hand, if $i \not\in \rho(X_j)$, then either $w_{\rho}[i]$ maps $X_j
\mapsto 0$, or the length of $w$ is smaller than or equal to $i$.
Notice that there exist an infinite number of encodings of $\rho$.
The shortest one is $w_{\rho}^s$ of the length $n+1$, where $n$ is the largest number appearing in any of the sets that is assigned to
a~variable of $\bbX$ in $\rho$, or $-1$ when all these sets are empty.
The rest of the encodings are all those corresponding to $w_{\rho}^s$ extended with an arbitrary number of
$\zerosymb$ symbols appended to its end.
% Given a word $w_{\rho}$ of a~length $m$ representing some assignment $\rho$,
% the $i$-th symbol of $w_{\rho}$, denoted as $w_{\rho}[i]$, says for every free
% variable $X_j \in \bbX$ of $\varphi$ whether the natural number $i$ is
% in $X_j$ (if $i < m$ and $X_j \mapsto 1$ in $w_{\rho}[i]$) or not (if $i \geq
% m$ or $X_j \mapsto 0$ in $w_{\rho}[i]$) in the assignment $\rho$ represented by
% $w_{\rho}$.
% For any assignment $\rho$ there exists an infinite number of its encodings, the
% shortest one and all the encodings obtained by appending an arbitrary number of
% $\zerosymb$ symbols to it.
For example,
\bintrack{X_1}{X_2}{0}{1},
\bintrack{X_1}{X_2}{00}{10},
\bintrack{X_1}{X_2}{000}{100},
\bintrack{X_1}{X_2}{000\dots 0}{100\dots 0}
are all encodings of the assignment $\rho = \left\{X_1 \mapsto \emptyset, X_2 \mapsto
\{0\}\right\}$.
For the soundness of the decision procedure, it is important that $\A_\varphi$
always accepts either all encodings of $\rho$ or none of them.

% We consider only \emph{well-formed encodings}, i.e.~encodings $w_{\rho}$ of
% length $m$ such that if $m > 0$, $w_{\rho}[m-1]$ does not map all variables to
% 0; this allows us to have for every assignment $\rho$ a unique encoding
% $w_{\rho}$.

%The FA $\A_\varphi$ is constructed according to the structure of the formula
%$\varphi$, starting with constructing FAs for the atomic formulae and then
%manipulating the FAs by operations given by the operators and quantifiers in
%the formula.

%The FAs for the atomic formulae are given in Fig.~\ref{fig:atomic_automata}; it
%is easy to see that they accept precisely all encodings of the models of the
%formulae.
%
%\begin{figure}[t]
%\begin{center}
%\td{OL: fill}
%\end{center}
%\caption{FAs for atomic formulae}
%\label{fig:atomic_automata}
%\end{figure}

%For the formulae of the form $\varphi \land \psi$, $\varphi \lor \psi$,
%we first obtain the FAs $\A_\varphi$ and $\A_\psi$.
%Having computed automata for formulae $\varphi$ and $\psi$,
The automata $\A_{\varphi \land \psi}$ and $\A_{\varphi \lor \psi}$ are constructed from $\A_\varphi$ and $\A_\psi$ by standard automata-theoretic union and intersection operations, preceded by the so-called cylindrification which unifies the alphabets of $\A_\varphi$ and $\A_\psi$.
%Because the free variables of $\varphi$ and $\psi$ may differ, the alphabets of
%$\A_\varphi$ and $\A_\psi$ may also differ.
%Hence the before computing the union/intersection,
%the alphabets are first synchronised using so-called cylindrification.
Since these operations, as well as the automata for the atomic formulae, are not the subject of the contribution %\td{OL: our method} 
proposed in this paper, we refer the interested reader to~\cite{tata} for details.

%The first step is therefore unification of the alphabets of $\A_\varphi$ and
%$\A_\psi$ so that the symbols of $\A_\varphi$ also contain assignments to the
%free variables of $\psi$ and the symbols of $\A_\psi$ also contain assignments
%to the free variables of $\varphi$.
%This is performed by computing the cylindrification of $\A_\varphi$ by the
%free variables of $\psi$ and the cylindrification of $\A_\psi$ by the free
%variables of $\varphi$.
%The FAs $\A_{\varphi \land \psi}$ and $\A_{\varphi \lor \psi}$ are then
%obtained as the intersection and union of these automata respectively.
%It is easy to see that $\A_{\varphi \land \psi}$ accepts encodings of
%assignments that are models of both $\varphi$ and $\psi$, and 
%$\A_{\varphi \lor \psi}$ accepts encodings of
%assignments that are models of either $\varphi$ or $\psi$.

The part of the procedure which is central for this paper is processing negation and existential quantification;
we will therefore describe it in detail.
The FA $\A_{\neg\varphi}$ is constructed as the complement of 
$\A_\varphi$.
Then, all encodings of the assignments that were accepted by $\A_\varphi$ are
rejected by $\A_{\neg\varphi}$ and vice versa.
The FA
$\A_{\exists X : \varphi}$ is obtained from the FA $\A_\varphi = (Q, \Delta, I,
\finst{})$ by first projecting
$X$ from the transition relation $\Delta$, yielding the FA
$\A_\varphi' = (Q, \projOfFrom{X}{\Delta}, I, \finst{})$.
However, $\A_\varphi'$ cannot be directly used as $\A_{\exists X : \varphi}$.
The reason is that $\A_\varphi'$ may now be inconsistent in accepting some
encodings of an assignment $\rho$ while rejecting other encodings of $\rho$.
For example, suppose that $\A_\varphi$ accepts the words 
\bintrack{X_1}{X_2}{010}{001},
\bintrack{X_1}{X_2}{0100}{0010},
\bintrack{X_1}{X_2}{0100\dots 0}{0010\dots 0}
and we are computing the FA for $\exists X_2: \varphi$.
When we remove the $X_2$ row from all symbols, we obtain the FA $\A_\varphi'$ that accepts the words
\unitrack{X_1}{010},
\unitrack{X_1}{0100},
\unitrack{X_1}{0100\dots 0},
but does not accept the word 
\unitrack{X_1}{01} that encodes the same assignment (because \bintrack{X_1}{X_2}{01}{??} $\not\in
\langof{A_\varphi}$ for any values in the places of~``?''s).
As a remedy for this situation, we need to modify $\A_\varphi'$ to also accept
the rest of the encodings of $\rho$.
This is done by enlarging the set of final states of $\A_\varphi'$ to also
contain all states that can reach a~final state of $\A_\varphi'$ by a~sequence
of $\zerosymb$ symbols.
%% %
%% %NOFIXPOINT VERSION
%% Formally,
%% $\A_{\exists X : \varphi} = (Q, \projOfFrom{X}{\Delta}, I, \finstex{})$
%% is obtained from
%% $\A_\varphi' = (Q, \projOfFrom{X}{\Delta}, I, \finst{})$
%% where $\finstex{}$ is the set of states $q$ such that $q \lrun{w} s\in\finst{}$ such that $w\in\zerosymb^*$. %
%
%FIXPOINT VERSION
Formally,
$\A_{\exists X : \varphi} = (Q, \projOfFrom{X}{\Delta}, I, \finstex{})$
is obtained from
$\A_\varphi' = (Q, \projOfFrom{X}{\Delta}, I, \finst{})$
by computing $\finstex{}$ from $\finst{}$ using the fixpoint computation
$
\finstex{} = \lfp Z\,.\, \finst{} \cup \preof{\projOfFrom{X}{\Delta}, \zerosymb}{Z}
$.
Intuitively, 
the least fixpoint denotes the set of states backward-reachable from $\finst{}$ following transitions of $\projOfFrom{X}{\Delta}$ labelled by $\zerosymb$.
% the least fixpoint denotes the set of states gathered by a backward search from $\finst{}$ following transitions of $\projOfFrom{X}{\Delta}$ labeled by $\zerosymb$.
%Because $\pre$ and $\cup$ are monotone, the least fixpoint exists and is
%computable (by the Kleene fixpoint theorem) as the supremum of the sequence
%\begin{align*}
%&\emptyset \\
%&F\\
%&F \cup \preof{\projOfFrom{X}{\Delta}, \zerosymb}{F}\\
%&F \cup \preof{\projOfFrom{X}{\Delta}, \zerosymb}{F} \cup \preNthof{\projOfFrom{X}{\Delta}, \zerosymb}{2}{F} \\
%&F \cup \preof{\projOfFrom{X}{\Delta}, \zerosymb}{F} \cup \preNthof{\projOfFrom{X}{\Delta}, \zerosymb}{2}{F} \cup \preNthof{\projOfFrom{X}{\Delta}, \zerosymb}{3}{F} \\
%\end{align*}
%
%\td{OL: short alternative}
%
%It is easy to see that the least fixpoint exists and can be computed.
%
%\td{OL: end of short alternative}

The procedure returns an automaton $\A_\varphi$ that accepts exactly all
encodings of the models of $\varphi$.
This means that the language of $\A_\varphi$ is
\begin{inparaenum}[(i)]
\item  universal iff $\varphi$ is valid,
\item  non-universal iff $\varphi$ is invalid,
\item  empty iff $\varphi$ is unsatisfiable, and
\item  non-empty iff $\varphi$ is satisfiable.
\end{inparaenum}
%The classical decision procedure therefore works by constructing the automaton
%$\A_\varphi$ and checking its language.
Notice that in the particular case of \emph{ground} formulae (i.e.\ formulae without free
variables), the language of $\A_\varphi$ is either
$\langof{\A_\varphi} = \{\zerosymb\}^{*}$ in the case $\varphi$ is valid, or
$\langof{\A_\varphi} = \emptyset$ in the case $\varphi$ is invalid.

%===============================================================================
\vspace{0mm}
\section{Nested Antichain-based Approach for Alternating Quantifiers}\label{sec:dec_proc_ws1s}
\vspace{0mm}
%===============================================================================

We now present our approach for dealing with alternating quantifiers in WS1S
formulae.
We consider a ground formula $\varphi$ of the form
\begin{equation}\label{eq:varphi}
\varphi =
\underset{\varphi_m}{\underbrace{
  \neg\,\exists \X_m\,
  \neg
  \underset{\hspace{-3mm}\iddots}{
    \dots 
      \neg \,\exists \X_2\,
      \underset{\varphi_1}{\underbrace{
        \neg \,\exists \X_1: \varphi_0(\bbX)
      }}
  }
}}
\end{equation}
where each %the leading negation $(\neg)$ is optional,
$\X_i$ is a set of variables $\{X_a, \dots, X_b\}$,
$\exists \X_i$  is an abbreviation for a~non-empty sequence
$\exists X_a \dots \exists X_b$ of consecutive existential
quantifications, and $\varphi_0$ is an arbitrary formula called the \emph{matrix} of $\varphi$.
Note that the problem of checking validity or satisfiability of a formula with
free variables can be easily reduced to this form.

The classical procedure presented in Section~\ref{sec:dec_proc} 
computes a sequence of automata $\A_{\varphi_0},\A_{\varphi_0^\sharp},
\dots, \A_{\varphi_{m-1}^\sharp},\A_{\varphi_m}$ where 
for all $0 \leq i \leq m-1$, 
$\varphi_{i}^\sharp = \exists \X_{i+1}:\varphi_{i}$ and
$\varphi_{i+1} = \neg \varphi_{i}^\sharp$. 
The $\varphi_i$'s are the subformulae of $\varphi$ shown in Equation~\ref{eq:varphi}.
Since eliminating existential quantification on the automata level introduces nondeterminism (due to the projection on the transition relation),
every $\A_{\varphi_i^{\sharp}}$ may be nondeterministic.
The computation of $\A_{\varphi_{i+1}}$ then involves subset construction and becomes exponential. 
The worst case complexity of eliminating the prefix is therefore the tower of exponentials of the height $m$. 
Even though the construction may be optimized, e.g.\ by
minimizing every $\A_{\varphi_i}$ (which is implemented by MONA), 
the size of the generated automata can quickly become intractable. 

The main idea of our algorithm is inspired by the antichain algorithms~\cite{doyen:antichain} for testing language universality of an automaton $\A$.
In a~nutshell, testing universality of $\A$ is testing whether in the complement $\overline \A$ of $\A$
(which is created by determinization via subset construction, followed by swapping final and non-final states),
an initial state can reach a final state.  
The crucial idea of the antichain algorithms is based on the following:
\begin{inparaenum}[(i)]
\item  The search can be done on-the-fly while constructing $\overline\A$.
\item  The sets of states that arise during the search are closed (upward or downward, depending on the variant of the algorithm).
\item  The computation can be done symbolically on the generators of these closed sets. 
\end{inparaenum}
It is enough to keep only the extreme generators of the closed sets (maximal for downward closed, minimal for upward closed).
The generators that are not extreme (we say that they are \emph{subsumed}) 
can be pruned away, 
which vastly reduces the search space. 
%The search algorithm exists in two basic variants: 
%it either computes the set of states reachable from the search origin (a standard greatest fixpoint computation)
%and whether the target is in the set,   
%or it computes the set of states that cannot reach the search origin (a least fixpoint computation)
%and test whether the target is outside the set.
%The generated sets of states are downward or upward closed depending on the variant.

We notice that individual steps of the algorithm for constructing $\A_\varphi$ are very similar to testing universality.
Automaton $\A_{\varphi_{i}}$ arises by subset construction from $\A_{\varphi_{i-1}^\sharp}$,
and to compute $\A_{\varphi_{i}^\sharp}$, 
it is necessary to compute the set of final states $\finstex{i}$. 
Those are states backward reachable from the final states of $\A_{\varphi_{i}}$ via a subset of transitions of $\Delta_i$ 
(those labelled by symbols projected to $\zerosymb$ by $\pi_{i+1})$.
To compute $F_i^\sharp$, the antichain algorithms could be actually taken off-the-shelf
and run with $\A_{\varphi_{i-1}^\sharp}$ in the role of the input $\A$ and $\A_{\varphi_{i}^\sharp}$ in the role of $\overline\A$.
However, this approach has the following two problems.
First, antichain algorithms do not produce the automaton $\overline \A$ (here $\A_{\varphi_{i}^\sharp}$), 
but only a symbolic representation of a set of (backward) reachable states (here of $F_i^\sharp$). 
Since $\A_{\varphi_{i}^\sharp}$ is the input of the construction of $\A_{\varphi_{i+1}}$, 
the construction of $\A_\varphi$ could not continue.
The other problem is that the size of the input $\A_{\varphi_{i-1}^\sharp}$ of the antichain algorithm
is only limited by the tower of exponentials of the height $i-1$, and this might be already far out of reach.

The main contribution of our paper is an algorithm 
that alleviates the two problems mentioned above. 
It is based on a novel way of performing not only one, 
but all the $2m$ steps of the construction of $\A_\varphi$ on-the-fly.
It uses a nested symbolic representation of sets of states and a form of nested subsumption pruning on 
all levels of their structure.
This is achieved by a substantial refinement of the basic ideas of antichain algorithms.

%In the following we use $\projNthOf{i}{C}$ for some object $C$ to denote the
%projection of all variables in $\X_i$ from the object $C$.
%\td{OL: formally?}

%*******************************************************************************
\vspace{-0.0mm}
\subsection{Structure of the Algorithm}\label{sec:structure}
\vspace{-0.0mm}
%*******************************************************************************

Let us now start explaining our on-the-fly algorithm for handling quantifier alternation.
Following the construction of automata described in Section~\ref{sec:dec_proc}, 
the structure of the automata from the previous section, $\A_{\varphi_0},\A_{\varphi_0^\sharp},
\dots, \A_{\varphi_{m-1}^\sharp},\A_{\varphi_m}$, can be described using the following recursive definition.
We use $\projNthOf{i}{C}$ for any mathematical structure $C$ to denote
projection of all variables in $\X_1 \cup \cdots \cup \X_i$ from $C$.
%\td{OL: we use $\A_i$ for $\A_{\varphi_i}$ and $\A_i^\sharp$ for $\A_{\varphi_i^\sharp}$}
%
% Let $\A_{\varphi_0} = (Q_{0}, \Delta_{0}, I_{0}, \finst{0})$ be an FA over $\bbX$.
% Then, for each $0\leq i < m$, 
% we define 
% $\A_{\varphi_i^\sharp} = (Q_{i}, \Delta_{i}^\sharp, I_{i}, \finst{i}^\sharp)$ to be an FA over   
% $\projNthOf{i+1}{\bbX}$ where 
% $\Delta_{i}^\sharp = \projNthOf{i+1}{\Delta_i}$
% and 
% $\finstex{i} = \lfp Z\,.\, \finst{i} \cup \preof{\Delta_i^\sharp, \zerosymb}{Z}$.
% The automaton 
% $\A_{\varphi_{i+1}} = (Q_{i+1}, \Delta_{i+1}, I_{i+1}, \finst{i+1})$
% is defined as an FA over $\projNthOf{i+1}{\bbX}$
% where 

Let $\A_{\varphi_0} = (Q_{0}, \Delta_{0}, I_{0}, \finst{0})$ be an FA over $\bbX$.
Then, for each $0\leq i < m$, 
$\A_{\varphi_i^\sharp}$ and $\A_{\varphi_{i+1}}$ are FAs over $\projNthOf{i+1}{\bbX}$
that have from the construction the following structure:
\begin{equation*}
\begin{array}{rlc@{~~~}|c@{~~~}rl}
\A_{\varphi_i^\sharp}             &             = (Q_{i}, \Delta_{i}^\sharp, I_{i}, \finst{i}^\sharp)~\mbox{where}      &&& \A_{\varphi_{i+1}}             &             = (Q_{i+1}, \Delta_{i+1}, I_{i+1}, \finst{i+1})~\mbox{where}\\
\scriptstyle\Delta_{i}^\sharp     &\scriptstyle = \projNthOf{i+1}{\Delta_i}~\mbox{and}                                  &&&\scriptstyle \Delta_{i+1}       &\scriptstyle = \left\{R \ltr{\tau} \postof{\Delta_i^\sharp, \tau}{R} \rlmid R\in Q_{i+1}\right\},\\
\scriptstyle\finstex{i}           &\scriptstyle = \lfp Z\,.\, \finst{i} \cup \preof{\Delta_i^\sharp, \zerosymb}{Z}.     &&&\scriptstyle Q_{i+1}            &\scriptstyle = 2^{Q_{i}}, \quad I_{i+1} = \{I_{i}\},\quad\mbox{and}\quad \finst{i+1} = \downclngen{\{Q_i\setminus\finstex{i}\}}.
\end{array}
\end{equation*}
%
%\begin{equation*}
%\A_{\varphi_i^\sharp} = (Q_{i}, \Delta_{i}^\sharp, I_{i}, \finst{i}^\sharp)\quad\quad
%\A_{\varphi_{i+1}} = (Q_{i+1}, \Delta_{i+1}, I_{i+1}, \finst{i+1})
%\end{equation*}
%%
%where $\Delta_{i}^\sharp = \projNthOf{i+1}{\Delta_i}$,
%$\finstex{i} = \lfp Z\,.\, \finst{i} \cup \preof{\Delta_i^\sharp, \zerosymb}{Z}$,
%and
%%
%\begin{flalign*}
%Q_{i+1}      &= 2^{Q_{i}},&\\
%\Delta_{i+1} &= \left\{R \ltr{\tau} \postof{\Delta_i^\sharp, \tau}{R} \rlmid R\in Q_{i+1}\right\},&\\
%I_{i+1}      &= \{I_{i}\},&\\
%\finst{i+1} &= \downclgen{Q_i\setminus\finstex{i}}.&
%\end{flalign*}
%
\noindent
We recall that $\A_{\varphi_i^\sharp}$ directly corresponds to existential quantification of the variable $X_i$ (cf.\ Section~\ref{sec:dec_proc}),
and $\A_{\varphi_{i+1}}$ directly corresponds to the complement of $\A_{\varphi_i^\sharp}$ (cf.\ Section~\ref{sec:fa}).

A crucial observation behind our approach is that,
because $\varphi$ is ground,
$\A_{\varphi}$ is an FA over an empty set of variables, and, therefore,
$\langof{\A_\varphi}$ is either the empty set $\emptyset$
or the set $\{\zerosymb\}^{*}$ (as described in Section~\ref{sec:dec_proc}).
Therefore, we need to distinguish between these two cases only.
To determine which of them holds, we do not need to
explicitly construct the automaton $\A_\varphi$.
Instead, it suffices to check
whether $\A_\varphi$ accepts the empty string $\epsilon$. 
This is equivalent to checking existence of a state that is at the same time final and initial, that is
\begin{equation}\label{eq:test}
\models \varphi \quad \mathrm{iff} \quad I_m \cap \finst{m} \neq \emptyset.
\end{equation}
To compute $I_m$ from $I_0$ is straightforward (it equals $\{\{\ldots\{\{I_0\}\}\ldots\}\}$ nested $m$-times). 
In the rest of the section, we will describe how to compute $\finst{m}$ (its symbolic representation), and how to test whether it intersects with $I_m$.

The algorithm takes advantage of the fact that to represent final states, one can use their complement, 
the set of non-final states.
For $0\leq i \leq m$, we write 
$\nfinst{i}$ and $\nfinst{i}^\sharp$ to denote the sets of non-final states $Q_i\setminus\finst{i}$ of $\A_i$ and $Q_i\setminus \finstex{i}$ of $\A_i^\sharp$ respectively.
The algorithm will then instead of computing the sequence of automata $\A_{\varphi_0}$, $\A_{\varphi_0^\sharp}$,
\dots, $\A_{\varphi_{m-1}^\sharp}$, $\A_{\varphi_m}$ compute the sequence 
$\finst{0}, \finstex{0}, \nfinst{1}, \nfinstex{1}, \ldots$ up to either
$\finst{m}$ (if $m$ is even) or $\nfinst{m}$ (if $m$ is odd),
which suffices for testing the validity of $\varphi$.
The algorithm starts with $\finst{0}$ and uses the following recursive equations:
\newcommand{\Fi}{i}
\newcommand{\Fisharp}{ii}
\newcommand{\Ni}{iii}
\newcommand{\Nisharp}{iv}
\begin{equation}
\begin{array}{rrl@{\ \ \ \ \ }rcl}\label{eq:fin_nfin_cl_sets}
\mbox{(\Fi)}      & \finst{i+1}  &= \downclgen{\nfinstex{i}},            & \mbox{(\Fisharp)} & \finstex{i}  &= \lfp Z\,.\, \finst{i} \cup \preof{\Delta_i^\sharp, \zerosymb}{Z}, \\[2mm]
\mbox{(\Ni)}      & \nfinst{i+1} &= \upclngen{\choice{\{\finstex{i}\}}}, & \mbox{(\Nisharp)} & \nfinstex{i} &= \gfp Z\,.\, \nfinst{i} \cap \cpreof{\Delta_i^\sharp, \zerosymb}{Z}.
\end{array}
\end{equation}
Intuitively, 
Equations~(\Fi) and (\Fisharp) are directly from the definition of $\A_i$ and $\A_i^\sharp$.
Equation~(\Ni) is a dual of Equation~(\Fi): 
$\nfinst{i+1}$ contains all subsets of $Q_i$ that contain at least one state from $\finstex{i}$
(cf.\ the definition of the $\choice{}$ operator).
Finally,
Equation~(\Nisharp) is a dual of Equation~(\Fisharp):
in the $k$-th iteration of the greatest fixpoint computation,
the current set of states $Z$ will contain all states that cannot reach an $\finst{i}$ state over $\zerosymb$ within $k$ steps.
In the next iteration, only those states of $Z$ are kept such that all their $\zerosymb$-successors are in $Z$.
Hence, the new value of $Z$ is the set of states that cannot reach $\finst{i}$ over $\zerosymb$ in $k+1$ steps, and the computation stabilises with the set of states that cannot reach $\finst{i}$ over $\zerosymb$ in any number of steps.

In the next two sections, we will show that both of the above fixpoint computations can be carried out symbolically on representatives of upward/downward closed sets. 
Particularly, in Sections~\ref{sec:cpre_to_pre} and \ref{sec:pre_to_cpre}, 
we show how the fixpoints from Equations~(\Fisharp{}) and (\Nisharp) can be computed symbolically, 
using subsets of $Q_{i-1}$ as representatives (generators) of upward/downward closed subsets of $Q_i$. 
Section~\ref{sec:terms} explains how the above symbolic fixpoint computations can be carried out using nested terms of depth $i$ as a~symbolic representation of computed states of $Q_i$. 
Section~\ref{sec:testing} shows how to test emptiness of $I_m \cap \finst{m}$ on the symbolic terms,
and Section~\ref{sec:subsumption} describes the subsumption relation used to minimize the symbolic term representation used within computations of Equations~(\Fisharp{}) and (\Nisharp{}).
Proofs of the lemmas and used equations can be found in Appendix~\ref{app:proofs}.

%\td{OL: think about the paragraph}
%The following observation about the two fixpoint computations (ii) and (iv) are the first steps towards an efficient algorithm:
%\begin{inparaenum}[(a)]
%\item
%Because $\A_{\varphi_i}$ is a result of a subset construction, 
%the transition relation $\projNthOf{i+1}{\Delta_i}$ is monotone wrt. $\subseteq$ as well as wrt. $\supseteq$. 
%Therefore, in (ii), the $\pre$ of a downward closed set is downward closed and in (iv), $\cpre$ of an upward closed set is upward closed.
%\item
%Intersection of upward closed sets is upward closed and the union of downward closed sets is downward closed.   
%\item
%Both fixpoint computations start from an upward closed set (ii), or from a downward closed set (iv).
%\end{inparaenum}
%The three observations (a), (b), and (c) put together imply that every intermediate value $Z$ computed by a step of the fixpoint computations (ii) and (iv) is downward closed for (ii) and upward closed for (iv).
%We may therefore use minimal/maximal generators of such closed sets as their symbolic representation, 
%without the need of enumerating explicitly all their elements.
%In the next sections, 
%we show how to efficiently compute steps of the fixpoint computations symbolically, 
%working only on their generators.\td{is it too fast and hand wavy?}
%\td{in the greatest fixp, we do not work with generators, but with a krygl representation of generators}

%*******************************************************************************
\vspace{-0.0mm}
\subsection{Computing $\nfinstex{i}$ on Representatives of $\upclngen{\choice \R}$-sets}
\label{sec:cpre_to_pre} 
\vspace{-0.0mm}
%*******************************************************************************

Computing $\nfinstex{i}$ at each odd level of the hierarchy of automata is done
by computing the greatest fixpoint of the function from Equation~\ref{eq:fin_nfin_cl_sets}(\Nisharp):
\begin{equation}
\fcnfinstex{i}(Z) = \nfinst{i} \cap \cpreof{\Delta_i^\sharp, \zerosymb}{Z}.
% \fcnfinstex{i} = \lambda Z.\, \nfinst{i} \cap \cpreof{\Delta_i^\sharp, \zerosymb}{Z}.
\end{equation}
We will show that the whole fixpoint computation from Equation~\ref{eq:fin_nfin_cl_sets}(\Nisharp) can be carried out symbolically
on the representatives of $Z$.
We will explain that:
\begin{inparaenum}[(a)]
\item  All intermediate values of $Z$ have the form $\upclngen{\choice \R}$, $\R\subseteq Q_i$, so the sets $\R$ can be used as their symbolic representatives.
\item  $\cpre$ and $\cap$ can be computed on such symbolic representation efficiently.
\end{inparaenum}

Let us start with the computation of $\cpreof{\Delta_i^\sharp, \tau}{Z}$ where $\tau\in \projNthOf{i+1}{\bbX}$, 
assuming that $Z$ is of the form $\upclngen{\choice \R}$, represented by $\R = \{R_1, \dots, R_n\}$. 
Observe that a set of symbolic representatives $\R$ %such as $\upclngen{\choice \R}$ 
stands for the intersection of denotations of individual representatives,
that is
\begin{equation}\label{eq:choice_intersection}
\upclngen{\choice{\R}}
= 
\bigcap_{R_j\in\R}\upclngen{\choice \{R_j\}} .
\end{equation}
$Z$ can thus be written as the $\cpre$-image $\cpreof{\Delta_i^\sharp,\tau}{\bigcap \mcS}$ of the intersection 
of the elements of a set $\mcS$ having the form $\upclngen{\choice \{R_j\}}, R_j\in\R$. 
Further, because $\cpre$ distributes over $\cap$, we can
compute the $\cpre$-image of an intersection by computing intersection of the $\cpre$-images, i.e.\ 
\begin{equation}\label{eq:cpre_intersection}
\cpreof{\Delta_i^\sharp,\tau}{\bigcap \mcS}
 = 
\bigcap_{S\in \mcS}\cpreof{\Delta_i^\sharp,\tau}{S}.
\end{equation}
By the definition of $\Delta_i^\sharp$ (where $\Delta_{i}^\sharp = \projNthOf{i+1}{\Delta_i}$),
$\cpreof{\Delta_i^\sharp,\tau}{S}$ 
can be computed using the transition relation $\Delta_i$
for the price of further refining the intersection. In particular,
\begin{equation}\label{eq:projection_intersection}
\cpreof{\Delta_i^\sharp, \tau}{S} 
=
\hspace{-3mm}\bigcap_{\omega\in\projInvOverOf{i+1}{\tau}} \hspace{-3mm}\cpreof{\Delta_i, \omega}{S}.
\end{equation}
Intuitively, 
$\cpreof{\Delta_i^\sharp, \tau}{S}$ contains states from which every transition labelled by \emph{any} symbol that is projected to $\tau$ by $\pi_{i+1}$ has its target in $S$.
Using Equations~\ref{eq:choice_intersection}, \ref{eq:cpre_intersection}, and~\ref{eq:projection_intersection}, 
we can write $\cpreof{\Delta_i^\sharp, \tau}{Z}$ as 
\begin{equation}
\bigcap_{\begin{array}{c} \\[-5.5mm] \scriptstyle S\in\mcS \\[-1mm] \scriptstyle \omega\in\projInvOverOf{i+1}{\tau} \end{array}} \hspace{-3mm}{\cpreof{\Delta_i, \omega}{S}}.
\end{equation}

To compute the individual conjuncts $\cpreof{\Delta_i, \omega}{S}$, 
we take advantage of the fact that every $S$ is in the special form $\upclngen{\choice \{R_j\}}$, 
and that $\Delta_i$ is, by its definition (obtained from determinization via subset construction), \emph{monotone} w.r.t. $\supseteq$. 
That is, if $P\ltr{\omega} P' \in \Delta_i$ for some $P,P'\in Q_i$, % and $\omega\in\Delta_i$,
then for every $R\supseteq P$, there is
$R'\supseteq P'$ s.t.~$R\ltr{\omega} R' \in \Delta_i$.
Due to monotonicity,
the $\cpre\scriptstyle{[\Delta_i, \omega]}$-image of an upward closed set is also upward closed.
Moreover, we observe that
it can be computed symbolically using $\pre$ on elements of its generators. 
Particularly, for a set of singletons ${S} = {\upclngen{\choice \{R_j\}}}$, we get the following equation:
\begin{equation}\label{eq:cpre_pre}
\cpreof
{\Delta_{i}, \omega}{\upclngen{\choice{\{R_j\}}}} = \upclngen{\choice{\left\{\preof{\Delta_{i-1}^\sharp,\omega}{R_j}\right\}}}.
\end{equation}
Intuitively, 
the sets with $\post$-images above a singleton $\{p\} \in \big\{\{p\} \mid p \in R_j\big\} = \upclngen{\choice{\{R_j\}}}$
are those that contain at least one state $q \in Q_{i-1}$ s.t.~$q \ltr{\omega} p \in \Delta_{i-1}^\sharp$.
Using Equation~\ref{eq:cpre_pre}, 
$\cpreof{\Delta_i^\sharp, \tau}{Z}$ can be rewritten as 
\begin{equation}
\bigcap_{\begin{array}{c} \\[-5.5mm] \scriptstyle R\in\R \\[-1mm] \scriptstyle \omega\in\projInvOverOf{i+1}{\tau} \end{array}} \hspace{-3mm} \upclngen{\choice{\left\{\preof{\Delta_{i-1}^\sharp,\omega}{R_j}\right\}}}.
\end{equation}
By applying Equation~\ref{eq:choice_intersection}, we get the final formula for $\cpreofred{\Delta_i^\sharp, \tau}$ shown in the lemma below.
\begin{lemma}\label{lemma:cpre}
$ 
{\cpreof{\Delta_i^\sharp, \tau}{\upclngen{\choice \R}}}
=
{\upclngen{\choice{\left\{\preof{{\Delta_{i-1}^\sharp},\omega}{R_j}\mid{\omega \in \projInvOverOf{i+1}{\tau},R_j \in \R}\right\}}}}.
$
\end{lemma}
In order to compute $\fcnfinstex{i}(Z)$, it remains to intersect $\cpreof{\Delta_i^\sharp, \zerosymb}{Z}$, computed using  
Lemma~\ref{lemma:cpre}, with 
$\nfinst{i}$. 
By Equation~\ref{eq:fin_nfin_cl_sets}(\Ni), $\nfinst{i}$ equals $\upclngen{\choice{\{\finstex{i-1}\}}}$,
and, by Equation~\ref{eq:choice_intersection}, the intersection can be done symbolically as
\begin{equation}\label{eq:Ni_step}
{\fcnfinstex{i}(Z)}
= 
{\upclngen{\choice{\left(
\{\finstex{i-1}\}
\cup
\left\{
\preof{{\Delta_{i-1}^\sharp},\omega}{R_j}\mid{\omega \in \projInvOverOf{i+1}{\zerosymb},R_j \in \R}
\right\}
\right)}}}.
\end{equation}
Finally, note that a
symbolic application of $\fcnfinstex{i}$ to $Z = \upclngen{\choice \R}$ represented as the set $\R$
reduces to computing $\pre$-images of the elements of $\R$, which are then put next to each other, together with $\finstex{i-1}$. 
The computation starts from $\nfinst{i} = \upclngen{\choice{\{\finstex{i-1}\}}}$, 
represented by $\{\finstex{i-1}\}$,
and each of its steps, implemented by Equation~\ref{eq:Ni_step}, 
preserves the form of sets $\upclngen{\choice \R}$, represented by $\R$.

\vspace{-0.0mm}
\subsection{Computing $\finstex{i}$ on Representatives of $\downclngen{\R}$-sets}\label{sec:pre_to_cpre}
\vspace{-0.0mm}
%*******************************************************************************

Similarly as in the previous section,
computation of $\finstex{i}$ at each even level of the automata hierarchy is
done by computing the least fixpoint of the function
\begin{equation}\label{eq:basic:finstex}
\fcfinstex{i}(Z) = \finst{i} \cup \preof{\Delta_i^\sharp, \zerosymb}{Z}.
% \fcfinstex{i} = \lambda Z.\,\finst{i} \cup \preof{\Delta_i^\sharp, \zerosymb}{Z}.
\end{equation}
We will show that the whole fixpoint computation from Equation~\ref{eq:fin_nfin_cl_sets}(\Fisharp) can be again carried out symbolically. 
We will explain the following:
(a) All intermediate values of $Z$ are of the form $\downclngen{\R}$, $\R\subseteq Q_i$,
so the sets $\R$ can be used as their symbolic representatives.
(b) $\pre$ and $\cup$ can be computed efficiently on such a~symbolic representation. 
The computation is a simpler analogy of the one in Section~\ref{sec:cpre_to_pre}.

We start with the computation of $\preof{\Delta_i^\sharp, \tau}{Z}$ 
where $\tau\in \projNthOf{i+1}{\bbX}$, 
assuming that $Z$ is of the form $\downclngen{\R}$, represented by $\R = \{R_1, \dots, R_n\}$. 
A simple analogy to Equations~\ref{eq:choice_intersection} and~\ref{eq:cpre_intersection} of Section~\ref{sec:cpre_to_pre} is that the union of downward closed sets is a downward closed set generated by the union of their generators,
i.e.\ $\downclngen\R = \bigcup_{R_j\in\R}\downclgen{R_j}$
and that $\pre$ distributes over union, i.e.\ 
\begin{equation}
\preof{\Delta_i^\sharp,\tau}{\bigcup \R}
 = 
 \bigcup_{R_j\in \R}\preof{\Delta_i^\sharp,\tau}{\downclgen{R_j}}.
\end{equation}
An analogy of Equation~\ref{eq:projection_intersection} holds too: 
%By the definition of $\Delta_i^\sharp$ (where $\Delta_{i}^\sharp = \projNthOf{i+1}{\Delta_i}$),
%the $\cpre\scriptstyle{[\Delta_i^\sharp,\tau]}(S)$ 
%can be computed using the transition relation $\Delta_i$,
%but for the price of further refining the intersection. 
%Particularly,
%
\begin{equation}\label{eq:projection_union}
\preof{\Delta_i^\sharp, \tau}{S} 
=
\hspace{-3mm}\bigcup_{\omega\in\projInvOverOf{i+1}{\tau}} \hspace{-3mm} \preof{\Delta_i, \omega}{S}.
\end{equation}
Intuitively, 
$\preof{\Delta_i^\sharp, \tau}{S}$ contains states
from which \emph{at least one} transition labelled by \emph{any} symbol that is projected to $\tau$ by $\pi_{i+1}$ leaves with the target in $S$.
Using Equation~\ref{eq:projection_union}, 
we can write $\preof{\Delta_i^\sharp, \tau}{Z}$ as 
\begin{equation}
\bigcup_{\begin{array}{c} \\[-5.5mm] \scriptstyle R_j\in\R \\[-1mm] \scriptstyle \omega\in\projInvOverOf{i+1}{\tau} \end{array}} \hspace{-3mm}{\preof{\Delta_i, \omega}{\downclgen{R_j}}}.
\end{equation}

To compute the individual disjuncts $\preof{\Delta_i, \omega}{\downclgen{R_j}}$, 
we take advantage of the fact that every $\downclgen{R_j}$ is downward closed, 
and that $\Delta_i$ is, by its definition (determinization by subset construction), \emph{monotone} w.r.t. $\subseteq$. 
That is, if $P\ltr{\omega} P' \in \Delta_i$ for some $P,P'\in Q_i$,
then for every $R\subseteq P$, there is
$R'\subseteq P'$ s.t.~$R\ltr{\omega} R' \in \Delta_i$.
Due to monotonicity,
the $\pre\scriptstyle{[\Delta_i, \omega]}$-image of a downward closed set is downward closed.
Moreover, we observe that it can be computed symbolically using $\cpre$ on elements of its generators. 
In particular, for a set $\downclgen{R_j}$, we get the following equation, which is a dual of Equation~\ref{eq:cpre_pre}:
\begin{equation}\label{eq:pre_cpre}
\preof
{\Delta_{i}, \omega}{\downclgen{R_j}} = \downclgen{\cpreof{\Delta_{i-1}^\sharp,\omega}{R_j}}.
\end{equation}
Intuitively, 
the sets with the $\post$-images below $R_j$ 
are those which do not have an outgoing transition leading outside $R_j$. The largest such set is $\cpreof{\Delta_{i-1}^\sharp,\omega}{R_j}$.
Using Equation~\ref{eq:pre_cpre},
$\preof{\Delta_i^\sharp, \tau}{Z}$ can be rewritten as 
\begin{equation}
\bigcup_{\begin{array}{c} \\[-5.5mm] \scriptstyle R_j\in\R \\[-1mm] \scriptstyle \omega\in\projInvOverOf{i+1}{\tau} \end{array}} \hspace{-3mm} \downclgen{\cpreof{\Delta_{i-1}^\sharp,\omega}{R_j}}
\end{equation}
which
gives us the final formula for $\preofred{\Delta_i^\sharp, \tau}$ described in Lemma~\ref{lemma:pre}.
\begin{lemma}\label{lemma:pre}
$ 
{\preof{\Delta_i^\sharp, \tau}{\downclngen \R}}
=
{\downclgen{\cpreof{{\Delta_{i-1}^\sharp},\omega}{R_j}\mid{\omega \in \projInvOverOf{i+1}{\tau},R_j \in \R}}}.
$
\end{lemma}
To compute $\fcfinstex{i}(Z)$, it remains to unite $\preof{\Delta_i^\sharp, \zerosymb}{Z}$, computed using  
Lemma~\ref{lemma:pre}, with 
$\finst{i}$. 
From Equation~\ref{eq:fin_nfin_cl_sets}(\Fi), $\finst{i}$ equals $\downclgen{\nfinstex{i-1}}$, so
the union can be done symbolically as
\begin{equation}\label{eq:Fi_step}
{\fcfinstex{i}(Z)}
= 
{\downclngen{\left(
\{\nfinstex{i-1}\}
\cup
\left\{
\cpreof{{\Delta_{i-1}^\sharp},\omega}{R_j}\mid{\omega \in \projInvOverOf{i+1}{\zerosymb},R_j \in \R}
\right\}
\right)}}.
\end{equation}
Therefore, 
a~symbolic application of $\fcfinstex{i}$ to $Z = \downclngen{\R}$ represented using the set $\R$
reduces to computing $\cpre$-images of elements of $\R$, which are put next to each other, together with $\nfinstex{i-1}$. 
The computation starts from $\finst{i} = \downclgen{\nfinstex{i-1}}$, 
represented by $\{\nfinstex{i-1}\}$,
and each of its steps, implemented by Equation~\ref{eq:Fi_step}, 
preserves the form of sets $\downclngen{\R}$, represented by $\R$.

\vspace{-0.0mm}
\subsection{Computation of $\finstex{i}$ and $\nfinstex{i}$ on Symbolic Terms}\label{sec:terms}
\vspace{-0.0mm}
%*******************************************************************************

\newcommand{\upkryglabc}[0]{\upclngen{\choice{\big\{\{a,b,c\}\big\}}}}
\newcommand{\upkryglabcbccd}[0]{\upclngen{\choice{\big\{\{a,b,c\}, \{b,c\}, \{c,d\}\big\}}}}
\newcommand{\upkryglabcbccdred}[0]{\upclngen{\choice{\big\{\{b,c\}, \{c,d\}\big\}}}}
\newcommand{\downupkryglabcbccd}[0]{\downclngen{\Big\{\!\upkryglabcbccd,\upclngen{\choice{\big\{\{b\},\{d\}\big\}}},\upclngen{\choice{\big\{\{a\},\{c,d\}\big\}}}\!\Big\}}}
\newcommand{\downupkryglabcbccdred}[0]{\downclngen{\Big\{\!\upkryglabcbccdred,\upclngen{\choice{\big\{\{b\},\{d\}\big\}}},\upclngen{\choice{\big\{\{a\},\{c,d\}\big\}}}\!\Big\}}}
\newcommand{\downupkryglabcbccdredred}[0]{\downclngen{\Big\{\!\upclngen{\choice{\big\{\{b\},\{d\}\big\}}},\upclngen{\choice{\big\{\{a\},\{c,d\}\big\}}}\!\Big\}}}

Sections~\ref{sec:cpre_to_pre} and \ref{sec:pre_to_cpre} show how sets of states arising within the fixpoint computations from Equations~\ref{eq:fin_nfin_cl_sets}(\Fisharp{})
and~\ref{eq:fin_nfin_cl_sets}(\Nisharp{}) can be represented symbolically using representatives which are sets of states of the lower level.
The sets of states of the lower level will be again represented symbolically.
When computing the fixpoint of level $i$, we will work with nested symbolic representation of states of depth $i$.
Particularly, sets of states of $Q_k,0\leq k \leq i$, are represented by \emph{terms of level $k$}
where a~term of level $0$ is a subset of $Q_0$,
a~term of level $2j+1$, $j \geq 0$, is of the form $\upclngen{\choice\{{t_1,\ldots,t_n}\}}$ 
where $t_1,\ldots,t_n$ are terms of level $2j$,
and a~term of level $2j$, $j>0$, is of the form $\downclgen{t_1,\ldots,t_n}$ 
where $t_1,\ldots,t_n$ are terms of level $2j-1$. 

The computation of $\cpre$ and $\fcnfinstex{2j+1}$ on a term of level $2j+1$ and computation of $\pre$ and $\fcfinstex{2j}$ on a term of level $2j$
then becomes a recursive procedure that descends via the structure of the terms and produces again a term
of level $2j+1$ or $2j$ respectively.
In the case of $\cpre$ and $\fcnfinstex{2j+1}$ called on a~term of level $2j+1$, 
Lemma~\ref{lemma:cpre} reduces the computation
to a computation of $\pre$ on its sub-terms of level $2j$, 
which is again reduced by Lemma~\ref{lemma:pre} 
to a computation of $\cpre$ on terms of level $2j-1$, and so on until
the bottom level where the algorithm computes $\pre$ on the terms of level $0$ (subsets of $Q_0$).
The case of $\pre$ and $\fcfinstex{2j}$ called on a~term of level $2j$ is symmetrical.

\paragraph{Example.}
We will demonstrate the run of our algorithm on the following abstract example.
Consider a ground WS1S formula $\varphi = \neg\exists\X_3 \neg\exists\X_2
\neg\exists\X_1 : \varphi_0$ and an FA $\A_0 = (Q_0, \Delta_0, I_0 = \{a\},
\finst{0} = \{a,b\})$ that represents $\varphi_0$.
Recall that our method decides validity of $\varphi$ by computing symbolically the sequence of sets
$\finstex{0}, \nfinst{1}, \nfinstex{1}, \finst{2}, \finstex{2}, \nfinst{3}$,
each of them represented using a~symbolic term,
and then checks if $I_3 \cap \nfinst{3} \neq \emptyset$.
In the following paragraph, we will show how such a~sequence is computed and
interleave the description with examples of possible intermediate results.

The fixpoint computation from Equation~\ref{eq:fin_nfin_cl_sets}(\Fisharp) of the first set in the sequence, 
$\finstex{0}$, 
is an explicit computation of the set of states backward-reachable from $\finst{0}$
via $\zerosymb$ transitions of $\Delta_0^\sharp$.
It is done using Equation~\ref{eq:basic:finstex}, yielding, e.g.\ the term
\begin{align*}
\termof{\finstex{0}}   &= \finstex{0} = \{a, b, c\}.
\end{align*}
%
%equation~(\ref{eq:fin_nfin_cl_sets}ii). %directly on $\Delta_0^\sharp$.
% $\finstex{0} = \lfp Z\,.\, \{a, b\} \cup \preof{\projNthOf{1}{\Delta_0}, \zerosymb}{Z}$.
The fixpoint computation of $\nfinstex{1}$ from
Equation~\ref{eq:fin_nfin_cl_sets}(\Nisharp) is done symbolically.
It starts from the set $\nfinst{1}$ represented using Equation~\ref{eq:fin_nfin_cl_sets}(\Ni) as
the term $\termof{\nfinst{1}} = \upclngen{\choice{\big\{\{a,b,c\}\big\}}}$, and each of its
iterations is carried out using Equation~\ref{eq:Ni_step}. 
Equation~\ref{eq:Ni_step} transforms the problem of computing
$\cpreofred{\Delta_1,\omega'}$-image of a term into a~computation of a series of
$\preofred{\Delta_0^\sharp, \omega}$-images of its sub-terms, which is carried
out using Equation~\ref{eq:basic:finstex} in the
same way as when computing $\termof{\finstex{0}}$, ending with, e.g.\ the term
\begin{align*}
\termof{\nfinstex{1}}  &= \upkryglabcbccd.
\end{align*}
%
%Let $\termof{\nfinstex{1}}$ be the term computed by (\ref{eq:fin_nfin_cl_sets}\Nisharp) representing $\nfinstex{1}$.
%$\cpreof{\projNthOf{2}{\Delta_1}, \zerosymb}{Z}$
%is computed using
%$\preof{\projNthOf{1}{\Delta_0}, \omega}{z}$
%on $\projNthOf{1}{\Delta_0}$, as shown in Section~\ref{sec:cpre_to_pre}.
The term representing $\finst{2}$ is then $\termof{\finst{2}} =
\downclgen{\termof{\nfinstex{1}}}$, due to
Equation~\ref{eq:fin_nfin_cl_sets}(\Fi).
The symbolic fixpoint computation of $\finstex{2}$ from
Equation~\ref{eq:fin_nfin_cl_sets}(\Fisharp) then starts from
$\termof{\finst{2}}$, in our example
\begin{align*}
\termof{\finst{2}  }   &= \downclngen{\Big\{\!\upkryglabcbccd\!\Big\}}.
\end{align*}
Its steps are computed using Equation~\ref{eq:Fi_step}, 
which transforms 
%computing $\cpre\scriptstyle{[\Delta_2^\sharp]}$ 
the computation of the image of $\preofred{\Delta_2^\sharp,\omega''}$ into computations of a series of 
$\cpreofred{\Delta_1^\sharp,\omega'}$-images of sub-terms. 
These are in turn transformed by Lemma~\ref{lemma:cpre} into computations of
$\preofred{\Delta_0^\sharp,\omega}$-images of sub-sub-terms, subsets of $Q_0$,
in our example yielding, e.g.\ the term
\begin{align*}
\termof{\finstex{2}}   &= \downupkryglabcbccd.
\end{align*}
%
%Let $\termof{\finstex{2}}$ be the term computed by Equation~(\ref{eq:fin_nfin_cl_sets}\Fisharp{}) as the representation of $\finstex{2}$.
%
Using Equation~\ref{eq:fin_nfin_cl_sets}(\Nisharp), the final term representing $\nfinst{3}$ is then%$\termof{\nfinst{3}} = {\upclngen{\choice{\{\termof{\finstex{2}}\}}}}$. 
\begin{align*}
\termof{\nfinst{3} }   &= \upclngen{\choice{\bigg\{\!\downupkryglabcbccd\!\bigg\}}}.
\end{align*}
In the next section, we will describe how we check whether $I_3 \cap \finst{3}
\neq \emptyset$ using the computed term $\termof{\nfinst{3}}$.
\vspace{-0.0mm}
\subsection{Testing $I_m \cap \finst{m} \stackrel{?}{\neq} \emptyset$ on Symbolic Terms}
\label{sec:testing}
\vspace{-0.0mm}
%*******************************************************************************

Due to the special form of the set $I_m$ (every $I_i,1\leq i \leq m$, is the
singleton set $\{I_{i-1}\}$, cf.\ Section~\ref{sec:structure}), 
the test $I_m \cap \finst{m} \neq \emptyset$ can be done efficiently over the 
symbolic terms representing $\finst{m}$. 
Because $I_m = \{I_{m-1}\}$ is a singleton set, testing $I_m \cap F_m \neq
\emptyset$ is equivalent to testing $I_{m-1} \in F_m$.
If $m$ is odd, our approach computes the symbolic representation of $\nfinst{m}$ instead of $\finst{m}$.
Obviously, since $\nfinst{m}$ is the complement of $\finst{m}$, it holds that $I_{m-1} \in \finst{m} \iff I_{m-1} \not\in \nfinst{m}$.
% testing $I_m \cap \finst{m} \neq \emptyset$ is equivalent to testing 
% $I_m \cap \nfinst{m} = \emptyset$.
%To test whether $I_m \cap \finst{m} \neq \emptyset$ with $\finst{m}$ given by
%its symbolic representative, we use the fact that for $1 \leq i \leq m$, $I_i$
%is of the form $I_i = \{I_{i-1}\}$.
% We will describe a way of testing $I_m \cap X = \emptyset$ for any set represented by a symbolic term $\termof X$, which can be used in both variants of the test, when $m$ is odd as well as when it is even. 
Our way of testing $I_{m-1} \in Y_m$ on a~symbolic representation of the set $Y_m$ of level $m$ is based on the following equations:
% Since for every $i>0$, $I_i = \{I_{i-1}\}$, it holds that $I_m \cap \finst{m} \neq \emptyset \iff I_{m-1} \in \finst{m} \iff I_{m-1} \not\in \nfinst{m}$. Further, for every set of sets $\bbY$,
%We can take further advantage of the structure of $I_i$ and test the membership
%$I_{m-1} \in \finst{m}$ using the recursive formulae
%
\begin{align}
\label{eq:membership_down}
                     &&&&&&\{x\} &\in \downclngen{\bbY}        &\iff\quad& \exists Y \in \bbY : x \in Y                    &&&&&&\\
\label{eq:membership_up}
                     &&&&&&\{x\} &\in \upclngen{\choice{\bbY}} &\iff\quad& \forall Y \in \bbY : x \in Y                    &&&&&&\\[1.5mm]
\label{eq:membership_base}
\text{and for $i=0$,}&&&&&&I_0   &\in \upclngen{\choice{\bbY}} &\iff\quad& \forall Y \in \bbY : I_0 \cap Y \neq \emptyset. &&&&&&
\end{align}

Given a symbolic term $\termof X$ of level $m$ representing a set $X\subseteq Q_m$,
testing emptiness of $I_m \cap \finst{m}$ or  $I_m \cap \nfinst{m}$ can be done over $\termof X$ by a recursive 
procedure that descends along the structure of $\termof X$ using Equations~\ref{eq:membership_down} and~\ref{eq:membership_up},
essentially generating an AND-OR tree,
terminating the descent by the use of Equation~\ref{eq:membership_base}.

\paragraph{Example.}
In the example of Section~\ref{sec:terms}, 
we would test whether $\{\{\{\{a\}\}\}\} \cap \nfinst 3 = \emptyset$ over $\termof{\nfinst 3}$.
This is equivalent to testing whether $I_2 = \{\{\{a\}\}\} \in \nfinst{3}$.
From Equation~\ref{eq:membership_up} we get
that
\begin{equation}
I_2 \in \nfinst{3} \iff I_1 = \{\{a\}\} \in \finstex{2}
\end{equation}
because
$\finstex{2}$ is the denotation of the only sub-term
$\termof{\finstex 2}$ of $\termof{\nfinst 3}$. 
Equation~\ref{eq:membership_down} implies that
\begin{equation}
I_1 =\{\{a\}\} \in \finstex{2}
\iff
\{a\} \in \nfinstex{1} \lor
\{a\} \in \upclngen{\choice{\big\{\{b\}, \{d\}\big\}}} \lor
\{a\} \in \upclngen{\choice{\big\{\{a\}, \{c, d\}\big\}}}.
\end{equation}
Each of the disjuncts could then be further reduced by
Equation~\ref{eq:membership_up} into a~conjunction of membership queries on
the base level which would be solved by Equation~\ref{eq:membership_base}.
Since none of the disjuncts is satisfied, we conclude
that $I_1 \not\in \finstex{2}$, so
$I_2 \not\in \nfinst{3}$, implying that $I_2 \in \finst{3}$ and thus obtain the result
$\models \varphi$.

%%% PUVODNI VERZE
%%% To test whether $I_m \cap \finst{m} \neq \emptyset$ with $\finst{m}$ given by
%%% its symbolic representative, we use the fact that for $1 \leq i \leq m$, $I_i$
%%% is of the form $I_i = \{I_{i-1}\}$.
%%% Therefore, it holds that $I_m \cap \finst{m} \neq \emptyset \iff I_{m-1} \in \finst{m}$.
%%% We can take further advantage of the structure of $I_i$ and test the membership
%%% $I_{m-1} \in \finst{m}$ using the recursive formulae
%%% \begin{align}\label{eq:membership}
%%% &&&&&&\{x\} &\in \downclngen{\bbY}        &\iff&& \exists Y \in \bbY : x \in Y &&&&&&\\
%%% &&&&&&\{x\} &\in \upclngen{\choice{\bbY}} &\iff&& \forall Y \in \bbY : x \in Y &&&&&&
%%% \end{align}
%%% with the following formula for the base level
%%% \begin{equation}
%%% I_0 \in \upclngen{\choice{\bbY}} \iff \forall Y \in \bbY : I_0 \cap Y \neq \emptyset.
%%% \end{equation}
%%% Note that for an odd $m$, when $\nfinst{m}$ is computed, it holds that $I_{m-1}
%%% \in \finst{m} \iff I_{m-1} \not\in \nfinst{m}$.

%*******************************************************************************
\vspace{-0.0mm}
\subsection{Subsumption of Symbolic Terms}\label{sec:subsumption}
\vspace{-0.0mm}
%*******************************************************************************

%\td{OL: before this, talk about the terms}
%\td{OL: or shall we just say "symbolic representation" and "symbolic representatives"?}

Although the use of symbolic terms instead of an explicit enumeration of sets of states
itself considerably reduces the searched space, 
an even greater degree
of reduction can be obtained using subsumption inside the symbolic
representatives to reduce their size, similarly as in the antichain algorithms \cite{wulf:antichains}.
For any set of sets $\bbX$ containing a pair of distinct
elements $Y, Z \in \bbX$ s.t.~$Y \subseteq Z$, it holds that
\begin{equation}\label{eq:subsumption}
\downclngen{\bbX} = \downclngen{(\bbX \setminus Y)}
\quad\mbox{and}\quad
\upclngen{\choice{\bbX}} = \upclngen{\choice{(\bbX \setminus Z)}}.
\end{equation}
Therefore, if $\bbX$ is used to represent the set $\downclngen{\bbX}$, the element $Y$ is \emph{subsumed} by $Z$ and can be removed from $\bbX$ without changing its denotation. 
Likewise, if $\bbX$ is used to represent $\upclngen{\choice{\bbX}}$, 
the element $Z$ is \emph{subsumed} by $Y$ and can be removed from $\bbX$ without changing its denotation.
We can thus simplify any symbolic term by pruning out its sub-terms that represent elements subsumed by elements represented by other sub-terms, without changing the denotation of the term.

Computing subsumption on terms can be done 
%set inclusion on our symbolic representation for a pair of sets of
%sets $\bbX$ and $\bbY$ can be done efficiently 
using the following two equations:
\begin{align}
\label{eq:subs_down}
&&&&&&\downclngen{\bbX}        &\subseteq \downclngen{\bbY}         &&\iff \quad \forall X \in \bbX \exists Y \in \bbY: X \subseteq Y &&&&\\[1mm]
\label{eq:subs_up}
&&&&&&\upclngen{\choice{\bbX}} &\subseteq \upclngen{\choice{\bbY}}  &&\iff \quad \forall Y \in \bbY \exists X \in \bbX: X \subseteq Y.&&&&
\end{align}
Using Equations~\ref{eq:subs_down} and~\ref{eq:subs_up}, testing subsumption of terms of level $i$ reduces to testing subsumption of terms of level $i-1$. 
The procedure for testing subsumption of two terms descends along the structure of the term, 
using Equations~\ref{eq:subs_down} and~\ref{eq:subs_up} on levels greater than $0$, 
and on level $0$, 
where terms are subsets of $Q_0$, 
it tests subsumption by set inclusion.

% \td{OL: say that we can limit the depth?}
%
% During the fixpoint computations according to Equations~(\ref{eq:fin_nfin_cl_sets}\Fisharp) and (\ref{eq:fin_nfin_cl_sets}\Nisharp), 
% we eagerly reduce the terms;
% every time we add an element to a set on any level, we check whether is is not
% subsumed by another element in the set and, in case it is not, we check
% whether it subsumes other elements in the set, which are then removed.

\paragraph{Example.}
In the example from Section~\ref{sec:terms}, we can use the inclusion
$\{b,c\} \subseteq \{a,b,c\}$ and Equation~\ref{eq:subsumption} to reduce
$\termof{\nfinstex{1}} = \upkryglabcbccd$ to the term
\begin{align*}
\termof{\nfinst{1}}'   &= \upkryglabcbccdred.
\end{align*}
Moreover, Equation~\ref{eq:subs_up} implies that 
$\upclngen{\choice{\big\{\{b,c\}, \{c,d\}\big\}}}$ is subsumed by the term
$\upclngen{\choice{\big\{\{b\}, \{d\}\big\}}}$, and, therefore, we can reduce
the term
$\termof{\finstex{2}}$ to the term
\begin{align*}
\termof{\finstex{2}}'   &= \downupkryglabcbccdredred.
\end{align*}
\vspace{-0.0mm}
\section{Experimental Evaluation}\label{sec:experiments}
\vspace{-0.0mm}
%%%%%%%%%%%%%%%%%%%%%%%%%%%%%%%%%%%%%%%%%%%%%%%%%%%%%%%%%%%%%%%%%%%%%%%%%%%%%%%%

\newcommand{\mdMona}[0]{I}
\newcommand{\mdVata}[0]{II}

We implemented a prototype of the presented approach in the tool
\texttt{dWiNA}~\cite{dWiNA} and evaluated it in a benchmark of both practical
and generated examples.
The tool uses the frontend of MONA to parse input formulae and also for the
construction of the base automaton $\A_{\varphi_0}$, and
further uses the MTBDD-based representation of FAs from the 
\texttt{libvata}~\cite{VATA} library.
The tool supports the following two modes of operation.

In mode~\mdMona, we use MONA to generate the deterministic automaton
$\A_{\varphi_0}$ corresponding to the matrix of the formula $\varphi$,
translate it to \texttt{libvata} and run our algorithm for handling the
prefix of $\varphi$ using \texttt{libvata}.
In mode~\mdVata, we first translate the formula $\varphi$ into the formula
$\varphi'$ in prenex normal form (i.e.\ it consists of a quantifier prefix and
a~quantifier-free matrix) where the occurence of negation in the matrix is limited to literals, and then
construct the nondeterministic automaton $\A_{\varphi_0}$ directly using
\texttt{libvata}.

% We implemented the procedure above as an extension of the tool
% \texttt{dWiNA}~\cite{dWiNA}\,---\,a prototype implementation of forward \wsks{}
% antichain-based decision procedure using non-deterministic automata. This tool
% uses the frontend of the MONA tool for parsing input formulae into intermediate
% representation which is further transformed to \epnf{} form. The base automaton
% of formulae can either be built-up by the library 
% or by MONA exploiting its many heuristics. Both of these approaches are using
% symbolic representation of transitions as MTBDDs.

\begin{table}[t]
\begin{center}
\caption{Results for practical examples}\label{tab:real}
\begin{tabular}{| l || r | r || r | r |}
 \hline
 \multicolumn{1}{| c ||}{\multirow{2}{*}{\textbf{Benchmark}}} &
 \multicolumn{2}{| c ||}{\textbf{Time [s]}} & \multicolumn{2}{| c |}{\textbf{Space [states]}}\\
 \cline{2-3}\cline{4-5}
   & \texttt{MONA} & \texttt{dWiNA} &\multicolumn{1}{c|}{~~\texttt{MONA}~~} & \multicolumn{1}{c|}{\texttt{dWiNA}} \\
  \hline
  \hline
  \texttt{reverse-before-loop} & 0.01 & 0.01 & 179 & 47 \\
  \hline
  \texttt{insert-in-loop} & 0.01 & 0.01 & 463 & 110 \\
  \hline
  \texttt{bubblesort-else} & 0.01 & 0.01 & 1\,285 & 271 \\
  \hline
  \texttt{reverse-in-loop} & 0.02 & 0.02 & 1\,311 & 274 \\
  \hline
  \texttt{bubblesort-if-else} & 0.02 & 0.23 & 4\,260 & 1\,040 \\
  \hline
  \texttt{bubblesort-if-if} & 0.12 & 1.14 & 8\,390 & 2\,065 \\
  \hline
  % \hline
  % \multicolumn{5}{ l |}{\textbf{Overall average:}} & \textbf{23.67}\,\% & \textbf{74.94}\,\%\\
  % \cline{6-7}
\end{tabular}
\end{center}
\end{table}

Our experiments were performed on an Intel Core i7-4770@3.4\,GHz processor with
32\,GiB RAM.
The practical formulae for our experiments that we report on here were obtained
from the shape analysis of~\cite{strand2}
and evaluated using mode~\mdMona{} of our tool;
the results are shown in Table~\ref{tab:real}
(see~\cite{dWiNA} for additional
experimental results).
We measure the time of runs of the tools for processing only the prefix of the
formulae.
We can observe that w.r.t.~the speed, we get comparable results; in some cases
\texttt{dWiNA} is slower than MONA, which we attribute to the fact that our prototype
implementation is, when compared with MONA, quite immature.
Regarding space, we compare the sum of the number of states of all automata
generated by MONA when processing the prefix of $\varphi$ with the number of
symbolic terms generated by \texttt{dWiNA} for processing the same.
We can observe a~significant reduction in the generated state space.
We also tried to run \texttt{dWiNA} on the modified formulae in mode \mdVata{}
but ran into the problem that we were not able to construct the
nondeterministic automaton for the quantifier-free matrix $\varphi_0$ in
reasonable time.
This was because after transformation of $\varphi$ into prenex normal form, if
$\varphi_0$ contains many conjunctions, the sizes of the automata generated
using intersection grow too large (one of the reasons for this is that
\texttt{libvata} in its current version does not support efficient reduction of
automata).

\begin{table}[t]
\begin{center}
 \caption{Results for generated formulae}
 \begin{tabular}{| c || r | r || r | r |}
 \hline
  & \multicolumn{2}{| c
 ||}{\textbf{Time [s]}} & \multicolumn{2}{| c |}{\textbf{Space [states]}}\\
 \cline{2-3}\cline{4-5}
    ~~~~$k$~~~~& \multicolumn{1}{c|}{\texttt{MONA}} & \multicolumn{1}{c||}{\texttt{dWiNA}} & \multicolumn{1}{c|}{\texttt{MONA}} & \multicolumn{1}{c|}{\texttt{dWiNA}} \\
  \hline
  \hline
 % 1 & 0.11 & 0.01 & 10\,718 & 39 & 0.36\% & 0.00\% \\
 % \hline
  2 & 0.20 & 0.01 & 25\,517 & 44 \\
  \hline
  3 & 0.57 & 0.01 & 60\,924 & 50 \\
  \hline
  4 & 1.79 & 0.02 & 145\,765 & 58 \\
  \hline
  5 & 4.98 & 0.02 & 349\,314 & 70 \\
  \hline
  6 & $\infty$ & 0.47 & $\infty$ & 90 \\
  \hline
\end{tabular}
\end{center}
\label{tab:gen}
\end{table}

To better evaluate the scalability of our approach, we created several
parameterized families of WS1S formulae. We start with basic formulae encoding interesting relations
among subsets of $\nat_0$, such as existence of certain transitive relations,
singleton sets, or intervals (their full definition can be found
in~\cite{dWiNA}).
From these we algorithmically create families of formulae with larger quantifier depth, 
regardless of the meaning of the created formulae (though their semantics is still nontrivial). 
In Table~\ref{tab:gen}, we give the results for one of the families where the basic formula expresses 
existence of an ascending chain of $n$ sets ordered
w.r.t.~$\subset$. The parameter $k$ stands for the number of alternations in the prefix of the formulae:
\begin{equation*} 
% Ascending chain condition
 \exists Y: \neg\exists X_1\neg\ldots\neg\exists X_k,\ldots,X_{n}:\!\!\!\bigwedge_{1
\leq i < n}\!\!\!\big(X_i \subseteq Y \wedge X_i \subset X_{i+1}\big) \Rightarrow
X_{i+1} \subseteq Y.
\label{eq:acc}
\end{equation*} 
We ran the experiments in mode \mdVata{} of \texttt{dWiNA} (the experiment in
mode \mdMona{} was not successful due to a too costly conversion of a large
base automaton from MONA to \texttt{libvata}).
\vspace{-0.0mm}
\section{Conclusion and Future Work}\label{sec:conclusion}
\vspace{-0.0mm}
%%%%%%%%%%%%%%%%%%%%%%%%%%%%%%%%%%%%%%%%%%%%%%%%%%%%%%%%%%%%%%%%%%%%%%%%%%%%%%%%

We presented a new approach for dealing with alternating quantifications within
the automata-based decision procedure for WS1S.
Our approach is based on a~generalization of the idea of the so-called antichain
algorithm for testing universality or language inclusion of finite automata.
Our approach processes a~prefix of the formula with an arbitrary number of
quantifier alternations on-the-fly using an efficient symbolic representation of
the state space, enhanced with subsumption pruning.
Our experimental results are encouraging (our tool often outperforms MONA) and
show that the direction started in this paper---using modern techniques for
nondeterministic automata in the context of deciding WS1S formulae---is
promising. 

An interesting direction of further development seems to be lifting the
symbolic $\pre$/$\cpre$ operators to a more general notion of terms that allow
working with general sub-formulae (that may include logical connectives and
nested quantifiers).
The algorithm could then be run over arbitrary formulae, without the need of
the transformation into the prenex form.
This would open a way of adopting optimizations used in other tools as well as
syntactical optimizations of the input formula such as anti-prenexing.
Another way of improvement is using simulation-based techniques to reduce the
generated automata as well as to weaken the term-subsumption relation (an
efficient algorithm for computing simulation over BDD-represented automata is
needed).
We also plan to extend the algorithms to WS$k$S and tree-automata, and perhaps
even further to more general inductive structures. 

\smallskip\noindent\emph{Acknowledgement.} The work in this technical report
was supported by the Czech
Science Foundation (projects 14-11384S and 202/13/37876P), the BUT FIT project
FIT-S-14-2486, and the EU/Czech IT4Innovations Centre of Excellence project
CZ.1.05/1.1.00/02.0070.

\bibliographystyle{splncs}
\bibliography{bibliography}

\clearpage
\vfill\pagebreak

\input{appendix}

\end{document}

%% file: macros.tex
\newcommand{\A}{\mathcal{A}}
\newcommand{\C}{\mathcal{C}}

\newcommand{\mcS}{\mathcal{S}}

\newcommand{\R}{\mathcal{R}}

\newcommand{\X}{\mathcal{X}}
\newcommand{\Y}{\mathcal{Y}}

\newcommand{\lang}{\mathcal{L}}
\newcommand{\langof}[1]{\lang(#1)}

\newcommand{\bbD}{\mathbb{D}}

\newcommand{\bbR}{\mathbb{R}}
\newcommand{\bbS}{\mathbb{S}}

\newcommand{\bbX}{\mathbb{X}}
\newcommand{\bbY}{\mathbb{Y}}

\newcommand{\nat}{\mathbb{N}}

\newcommand{\ltr}[1]{\xrightarrow{#1}} % labelled->

\newcommand{\run}{\Longrightarrow}
\newcommand{\lrun}[1]{\stackrel{#1}{\run}} % labelled->

\newcommand{\powerset}[1]{2^{#1}}
\newcommand{\st}{\;|\;}

\newcommand{\sing}{\mathrm{Sing}}
\newcommand{\singof}[1]{\sing(#1)}

\newcommand{\td}[1]{\textcolor{blue}{\ifmmode \text{[TODO: #1]}\else [TODO: #1] \fi}}

\newcommand{\projsymbol}[0]{\pi}
\newcommand{\projOfFrom}[2]{\projsymbol_{[#1]}(#2)}
\newcommand{\projNthOf}[2]{\projsymbol_{#1}(#2)}
\newcommand{\projInvOverOf}[2]{\projsymbol^{-1}_{#1}(#2)}

\newcommand{\zerosymb}[0]{\overline{0}}

\newcommand{\gfp}[0]{\nu}
\newcommand{\lfp}[0]{\mu}

\newcommand{\post}[0]{\mathit{post}}
\newcommand{\postof}[2]{\post{\scriptstyle[#1]}(#2)}
\newcommand{\pre}[0]{\mathit{pre}}
\newcommand{\preofred}[1]{\pre{\scriptstyle[#1]}}
\newcommand{\preof}[2]{\preofred{#1}(#2)}

\newcommand{\cpre}[0]{\mathit{cpre}}
\newcommand{\cpreofred}[1]{\cpre{\scriptstyle[#1]}}
\newcommand{\cpreof}[2]{\cpreofred{#1}(#2)}

\newcommand{\finst}[1]{F_{#1}}
\newcommand{\finstex}[1]{F^{\sharp}_{#1}}
\newcommand{\fcfinstex}[1]{f_{\finstex{#1}}}

\newcommand{\nfinst}[1]{N_{#1}}
\newcommand{\nfinstex}[1]{N^{\sharp}_{#1}}
\newcommand{\fcnfinstex}[1]{f_{\nfinstex{#1}}}

\newcommand{\downcl}[0]{\downarrow\!\!}
\newcommand{\downclgen}[1]{{\downcl\{#1\}}}
\newcommand{\downclngen}[1]{{\downcl\,#1}}
\newcommand{\upcl}[0]{\uparrow\!\!}

\newcommand{\upclngen}[1]{{\upcl\,#1}}

\newcommand{\relmiddle}[1]{\mathrel{}\middle#1\mathrel{}}
\newcommand{\rlmid}{\relmiddle{|}}

\newcommand{\choice}[1]{{\textstyle\coprod}#1}

\newcommand{\unitrack}[2]{{\scriptsize $\begin{array}{rl} #1:&#2\end{array}$}}
\newcommand{\bintrack}[4]{{\scriptsize $\begin{array}{rl} #1:&#3\\ #2:&#4\end{array}$}}

\newcommand{\termof}[1]{{t\scriptstyle{[#1]}}}

%% file: appendix.tex
\appendix

%%%%%%%%%%%%%%%%%%%%%%%%%%%%%%%%%%%%%%%%%%%%%%%%%%%%%%%%%%%%%%%%%%%%%%%%%%%%%%%%
\vspace{-0.0mm}
\section{Proofs for Section~\ref{sec:dec_proc_ws1s}}\label{app:proofs}
\vspace{-0.0mm}
%%%%%%%%%%%%%%%%%%%%%%%%%%%%%%%%%%%%%%%%%%%%%%%%%%%%%%%%%%%%%%%%%%%%%%%%%%%%%%%%

\begin{lemma}\label{lem:choice_intersection_pair}
Let $\X$ and $\Y$ be sets of sets.
Then it holds that
\begin{equation}
\upclngen{\choice{\bbX}} \cap \upclngen{\choice{\bbY}} = \upclngen{\choice{\left(\bbX \cup \bbY\right)}}.
\end{equation}
\end{lemma}

\begin{proof}
From the definition of the $\choice$ operator, it holds that
\begin{align}
\begin{split}
\upclngen{\choice{\bbX}} &= \upclngen{\big\{\{x_1, \dots, x_n\} \,\big|\, (x_1, \dots, x_n) \in \prod \bbX\big\}} \quad \mbox{and}\\
\upclngen{\choice{\bbY}} &= \upclngen{\big\{\{y_1, \dots, y_m\} \,\big|\, (y_1, \dots, y_m) \in \prod \bbY\big\}}.
\end{split}
\end{align}
Notice that the intersection of a pair of upward closed sets given by their
generators can be constructed by taking all pairs of generators $(X,Y)$, s.t.\ $X$ is from
$\choice{\bbX}$ and $Y$ is from $\choice{\bbY}$, and constructing the set $X \cup Y$.
It is easy to see that $X \cup Y$ is a generator of $\upclngen{\choice{\bbX}}
\cap \upclngen{\choice{\bbY}}$ and that $\upclngen{\choice{\bbX}}
\cap \upclngen{\choice{\bbY}}$ is generated by all such pairs, i.e.\ that
$\upclngen{\choice{\bbX}} \cap \upclngen{\choice{\bbY}}$ is equal to
\begin{equation}
\upclngen{\big\{\{x_1, \dots, x_n\} \cup  \{y_1, \dots, y_m\} \,\big|\, (x_1, \dots, x_n) \in \prod X \land (y_1, \dots, y_m) \in \prod Y\big\}}.
\end{equation}
We observe that this set can be also expressed as
\begin{equation}
\upclngen{\big\{\{x_1, \dots, x_n, y_1, \dots, y_m\} \,\big|\, (x_1, \dots, x_n, y_1, \dots y_m) \in \prod (X \cup Y)\big\}}
\end{equation}
or, to conclude the proof, as $\upclngen{\choice{\left(\bbX \cup \bbY\right)}}$.
\qed
\end{proof}

%%%%%%%%%%%%%%%%%%%%%%%%%%%%%%%%%%%%%%%%%%%%%%%%%%%%%%%%%%%%%%%%%%%%%%%%%%%%%%

\begin{lemma}
(Equation~\ref{eq:choice_intersection})
Let $\bbR$ be a set of sets.
Then, it holds that
\begin{equation}
\upclngen{\choice{\bbR}}
= 
\bigcap_{R_j\in\bbR}\upclngen{\choice \{R_j\}}.
\end{equation}
\end{lemma}

%%%%%%%%%%%%%%%%%%%%%%%%%%%%%%%%%%%%%%%%%%%%%%%%%%%%%%%%%%%%%%%%%%%%%%%%%%%%%%

\begin{proof}
Because intersection and union are both associative operations and
$\bbR = \{R_1, \dots, R_n\}$, this lemma is a simple consequence of
Lemma~\ref{lem:choice_intersection_pair}.
\qed
\end{proof}

%%%%%%%%%%%%%%%%%%%%%%%%%%%%%%%%%%%%%%%%%%%%%%%%%%%%%%%%%%%%%%%%%%%%%%%%%%%%%%

\begin{lemma}
(Equation~\ref{eq:cpre_pre})
Let $R_j \subseteq Q_{i-1}$ and $\omega$ be a symbol over $\projNthOf{i}{\bbX}$
for $i > 0$.
Then
\begin{equation}
\cpreof{\Delta_{i}, \omega}{\upclngen{\choice{\{R_j\}}}} = \upclngen{\choice{\left\{\preof{\Delta_{i-1}^\sharp,\omega}{R_j}\right\}}}.
\end{equation}
\end{lemma}

%%%%%%%%%%%%%%%%%%%%%%%%%%%%%%%%%%%%%%%%%%%%%%%%%%%%%%%%%%%%%%%%%%%%%%%%%%%%%%

\begin{proof}
First, we show that the set $\cpreof{\Delta_{i}, \omega}{\upclngen{\choice{\{R_j\}}}}$
is upward closed.
Second, we show that all elements of the set
$\choice{\left\{\preof{\Delta_{i-1}^\sharp,\omega}{R_j}\right\}}$
are contained in
$\cpreof{\Delta_{i}, \omega}{\upclngen{\choice{\{R_j\}}}}$.
Finally, we show that for every element $T$ in the set
$\cpreof{\Delta_{i}, \omega}{\upclngen{\choice{\{R_j\}}}}$ there is a smaller
element $S$ in the set
$\choice{\left\{\preof{\Delta_{i-1}^\sharp,\omega}{R_j}\right\}}$.

\begin{enumerate}
\item  Proving that $\cpreof{\Delta_{i}, \omega}{\upclngen{\choice{\{R_j\}}}}$
  is upward closed:
  Consider a state $S \in Q_i$ s.t.\ $S \in \cpreof{\Delta_{i}, \omega}{\upclngen{\choice{\{R_j\}}}}$.
  From the definition of $\cpre$, it holds that
  \begin{equation}
  \postof{\Delta_i, \omega}{\{S\}} \subseteq \upclngen{\choice{\{R_j\}}},
  \end{equation}
  and from the definition of $\Delta_i$, it holds that
  \begin{equation}
  \postof{\Delta_i, \omega}{\{S\}} = \{\postof{\Delta_{i-1}^\sharp, \omega}{S}\}.
  \end{equation}
  For $T \supseteq S$, it clearly holds that
  \begin{equation}
  \postof{\Delta_{i-1}^\sharp, \omega}{T} \supseteq \postof{\Delta_{i-1}^\sharp, \omega}{S}
  \end{equation}
  and, therefore, it also holds that
  \begin{equation}
  \postof{\Delta_i, \omega}{\{T\}} = \{\postof{\Delta_{i-1}^\sharp, \omega}{T}\} \subseteq \upclngen{\choice{\{R_j\}}}.
  \end{equation}
  Therefore,
  $T \in \cpreof{\Delta_{i}, \omega}{\upclngen{\choice{\{R_j\}}}}$
  and the set $\cpreof{\Delta_{i}, \omega}{\upclngen{\choice{\{R_j\}}}}$ is upward closed.

\item  Proving that for all elements
  $S \in \choice{\left\{\preof{\Delta_{i-1}^\sharp,\omega}{R_j}\right\}}$
  it holds that
  $S \in \cpreof{\Delta_{i}, \omega}{\upclngen{\choice{\{R_j\}}}}$:
  From the properties of $\choice$, it holds that $S = \{s\}$ is a~singleton.
  Because $s \in \preof{\Delta_{i-1}^\sharp,\omega}{R_j}$, there is a transition
  $s \ltr{\omega} r \in \Delta_{i-1}^\sharp$ for some $r \in R_j$.
  Since $\postof{\Delta_{i-1}^\sharp, \omega}{S} \supseteq \{r\}$, it
  follows from the definition of $\Delta_i$ that
  $\postof{\Delta_i, \omega}{\{S\}} = \{T\}$ where $T \supseteq \{r\}$,
  and so $T \in \upclngen{\choice{\{R_j\}}}$ and 
  $\postof{\Delta_i, \omega}{\{S\}} \subseteq \upclngen{\choice{\{R_j\}}}$.
  We use the definition of $\cpre$ to conclude that
  $S \in \cpreof{\Delta_{i}, \omega}{\upclngen{\choice{\{R_j\}}}}$.

\item  Proving that for every
  $T \in \cpreof{\Delta_{i}, \omega}{\upclngen{\choice{\{R_j\}}}}$
  there exists some element
  $S \in \choice{\left\{\preof{\Delta_{i-1}^\sharp,\omega}{R_j}\right\}}$
  such that $S \subseteq T$:
  From $T \in \cpreof{\Delta_{i}, \omega}{\upclngen{\choice{\{R_j\}}}}$ and
  the definition of $\Delta_i$, we have that
  \begin{equation}
  \postof{\Delta_i, \omega}{\{T\}} = \{P\} \subseteq \upclngen{\choice{\{R_j\}}}
  \end{equation}
  for $P$ s.t.~$\postof{\Delta_{i-1}^\sharp, \omega}{T} = P$.
  Since $P \in \upclngen{\choice{\{R_j\}}}$, there exists $r \in R_j \cap P$
  and $t \in T$ s.t.~$t \ltr{\omega} r \in \Delta_{i-1}^\sharp$.
  Because $t \in \preof{\Delta_{i-1}^\sharp, \omega}{\{r\}}$, we choose
  $S = \{t\}$ and we are done. 
\qed
\end{enumerate}
\end{proof}

%%%%%%%%%%%%%%%%%%%%%%%%%%%%%%%%%%%%%%%%%%%%%%%%%%%%%%%%%%%%%%%%%%%%%%%%%%%%%%

\begin{lemma}
(Equation \ref{eq:pre_cpre})
Let $R_j \subseteq Q_{i-1}$ and $\omega$ be a symbol over $\projNthOf{i}{\bbX}$
for $i > 0$.
Then
\begin{equation}
\preof
{\Delta_{i}, \omega}{\downclgen{R_j}} = \downclgen{\cpreof{\Delta_{i-1}^\sharp,\omega}{R_j}}.
\end{equation}
\end{lemma}

\begin{proof}
First, we show that $\preof{\Delta_{i}, \omega}{\downclgen{{R_j}}}$
is downward closed.
Second, we show that
$S = \cpreof{\Delta_{i-1}^\sharp,\omega}{R_j}$
is in
$\preof{\Delta_{i}, \omega}{\downclgen{R_j}}$.
Finally, we show that every element $T$ in
$\preof{\Delta_{i}, \omega}{\downclgen{R_j}}$ is smaller
than $S$.

\begin{enumerate}
\item  Proving that $\preof{\Delta_{i}, \omega}{\downclgen{R_j}}$
  is downward closed:
  Consider a state $S' \in Q_i$ s.t.\ $S' \in \preof{\Delta_{i}, \omega}{\downclgen{R_j}}$.
  From the definitions of $\pre$ and $\Delta_i$, it holds that
  \begin{equation}
  \postof{\Delta_i, \omega}{\{S'\}} = \{\postof{\Delta_{i-1}^\sharp, \omega}{S'}\} \subseteq \downclgen{R_j},
  \end{equation}
  and, therefore, $\postof{\Delta_{i-1}^\sharp, \omega}{S'} \in \downclgen{R_j}$.
  For $T \subseteq S'$, it clearly holds that
  \begin{equation}
  \postof{\Delta_{i-1}^\sharp, \omega}{T} \subseteq \postof{\Delta_{i-1}^\sharp, \omega}{S'}
  \end{equation}
  and so it also holds that
  \begin{equation}
  \postof{\Delta_i, \omega}{\{T\}} = \{\postof{\Delta_{i-1}^\sharp, \omega}{T}\} \subseteq \downclgen{R_j}.
  \end{equation}
  Therefore,
  $T \in \preof{\Delta_{i}, \omega}{\downclgen{R_j}}$
  and $\preof{\Delta_{i}, \omega}{\downclgen{R_j}}$ is downward closed.

\item  Proving that
  $S = \cpreof{\Delta_{i-1}^\sharp,\omega}{R_j} \in \preof{\Delta_{i}, \omega}{\downclgen{R_j}}$:
  From the definition of $\cpre$, it holds that
  \begin{equation}
  \postof{\Delta_{i-1}^\sharp, \omega}{S}  = S' \subseteq R_j.
  \end{equation}
  Further, from the definition of $\Delta_i$, it holds that
  $S \ltr{\omega} S' \in \Delta_i$ and, therefore, $S \in \preof{\Delta_i, \omega}{\downclgen{R_j}}$.

\item  Proving that for every
  $T \in \preof{\Delta_{i}, \omega}{\downclgen{R_j}}$
  it holds that $T \subseteq S$:
  From $T \in \preof{\Delta_{i}, \omega}{\downclgen{R_j}}$,
  we have that $T \ltr{\omega} P \in \Delta_{i}$
  for $P \subseteq R_j$, and, from the definition of $\Delta_i$, we have that
  $P = \postof{\Delta_{i-1}^\sharp,\omega}{T}$.
  From
  $P = \postof{\Delta_{i-1}^\sharp,\omega}{T}$
  and the definition of $\cpre$, it is easy to see that
  $T \subseteq \cpreof{\Delta_{i-1}^\sharp, \omega}{P}$, and, moreover
  \begin{equation}
  P \subseteq R_j\quad \implies \quad
  \cpreof{\Delta_{i-1}^\sharp, \omega}{P} \subseteq \cpreof{\Delta_{i-1}^\sharp, \omega}{R_j}.
  \end{equation}
  Therefore, we can conclude that $T \subseteq \cpreof{\Delta_{i-1}^\sharp, \omega}{R_j} = S$.
  \qed
\end{enumerate}
\end{proof}